\newif\ifdraft

\let\ifdraft\iffalse %

\ifdraft
    \documentclass[reqno]{amsart}						%
    \usepackage{microtype}
    \usepackage{fullpage}
\else

    \documentclass[smallcondensed]{svjour3}
    \smartqed  

    \makeatletter
    \def\cl@chapter{\@elt {theorem}}
    \makeatother
\fi

\ifdraft
\usepackage{amsmath}						%
\usepackage{amssymb}						%
\usepackage{amsfonts}						%

\usepackage{amsthm}

\usepackage[foot]{amsaddr}						%
\usepackage{array}

\usepackage{mathtools}						%
\mathtoolsset{%
showonlyrefs,	%
}
\else
\usepackage{amsmath}
\usepackage{amsthm}
\usepackage{amssymb}
\fi

\usepackage[utf8]{inputenc}							%
\usepackage[T1]{fontenc}						%
\DeclareUnicodeCharacter{00A0}{ }

\ifdraft
\usepackage[proportional,lining]{libertine}
\usepackage[libertine,libaltvw,cmintegrals,varbb]{newtxmath}
\fi

\usepackage{dsfont}							%

\usepackage[%
cal=cm,
]
{mathalfa}

\usepackage[labelfont={bf,small},labelsep=colon,font=small]{caption}	%
\usepackage[dvipsnames,svgnames]{xcolor}						%
\colorlet{MyBlue}{DodgerBlue!75!Black}
\colorlet{MyGreen}{DarkGreen!85!Black}

\usepackage{subcaption}
\usepackage{tikz}	\usetikzlibrary{patterns}
\usepackage{pgfplots}

\pgfplotsset{compat=1.14}
\usepgfplotslibrary{fillbetween}

\newenvironment{customlegend}[1][]{%
    \begingroup
    \csname pgfplots@init@cleared@structures\endcsname
    \pgfplotsset{#1}%
}{
    \csname pgfplots@createlegend\endcsname
    \endgroup
}
\def\addlegendimage{\csname pgfplots@addlegendimage\endcsname}

\usepackage[many]{tcolorbox}
\usepackage[algo2e]{algorithm2e}
\usepackage{algorithm}
\usepackage{array}
\usepackage{multicol}
\usepackage{multirow}

\newtcolorbox{worker}[1]{%
    tikznode boxed title,
    boxsep = -1mm,
    enhanced,
    arc=0mm,
    interior style={white},
    boxrule=0.5mm,
    width=0.95\linewidth,
    attach boxed title to top center= {yshift=-\tcboxedtitleheight/2},
    fonttitle=\bfseries,
    colbacktitle=white,coltitle=black,
    boxed title style={size=normal,colframe=white,boxrule=0pt},
    title={#1}
    }

\usepackage{acronym}						%
\usepackage{booktabs}						%
\usepackage{latexsym}						%
\usepackage{paralist}						%
\usepackage{xspace}							%

\ifdraft
\usepackage[sort&compress]{natbib}					%
\def\bibfont{\footnotesize}

\def\bibhang{24pt}

\fi

\usepackage{hyperref}
\hypersetup{
final,
colorlinks=true,
linktocpage=true,
pdfstartview=FitH,
breaklinks=true,
pdfpagemode=UseNone,
pageanchor=true,
pdfpagemode=UseOutlines,
plainpages=false,
bookmarksnumbered,
bookmarksopen=false,
bookmarksopenlevel=1,
hypertexnames=true,
pdfhighlight=/O,
urlcolor=Maroon,linkcolor=MyBlue!60!black,citecolor=DarkGreen!70!black,	%
pdftitle={},
pdfauthor={},
pdfsubject={},
pdfkeywords={},
pdfcreator={pdfLaTeX},
pdfproducer={LaTeX with hyperref}
}

\numberwithin{equation}{section}						%
\usepackage[sort&compress,capitalize,nameinlink]{cleveref}			%
\crefrangeformat{equation}{\upshape(#3#1#4)\textendash(#5#2#6)}

\ifdraft
    \theoremstyle{plain}
    \newtheorem{theorem}{Theorem}						%
    \newtheorem*{theorem*}{Theorem}						%
    \newtheorem*{corollary*}{Corollary}						%
    \newtheorem{lemma}[theorem]{Lemma}					%
    \newtheorem{proposition}[theorem]{Proposition}				%

    \theoremstyle{definition}
    \newtheorem*{definition*}{Definition}					%
    \newtheorem{assumption}{Assumption}					%
    \newtheorem*{assumption*}{Assumptions}					%

    \theoremstyle{remark}
    \newtheorem{remark}{Remark}						%
    \newtheorem*{remark*}{Remark}						%
    \newtheorem*{example*}{Example}						%
\else
    \newtheorem{assumption}{Assumption}
\fi

\ifdraft
\fi

\numberwithin{equation}{section}						%
\numberwithin{theorem}{section}						%
\numberwithin{remark}{section}						%
\numberwithin{example}{section}						%

\usepackage[scaled=.75]{beramono}
\usepackage{listings}

\lstdefinelanguage{Julia}%
  {morekeywords={abstract,break,case,catch,const,continue,do,else,elseif,%
      end,export,false,for,function,immutable,import,importall,if,in,%
      macro,module,otherwise,quote,return,struct,switch,true,try,type,typealias,%
      using,while},%
    sensitive=true,%
    morecomment=[l]\#,%
    morecomment=[n]{\#=}{=\#},%
    morestring=[s]{"}{"},%
    morestring=[m]{'}{'},%
}[keywords,comments,strings]%

\floatstyle{ruled}
\newfloat{Algorithm}{h!}{lop}%

\newcommand{\NA}{\mathcal{W}}
\newcommand{\A}{\mathsf{A}}
\newcommand{\B}{\mathsf{B}}
\newcommand{\M}{\mathsf{M}}
\newcommand{\E}{\RR^{S\times n}}

\newcommand{\selected}{s^k}

\DeclareMathOperator{\var}{x} %
\DeclareMathOperator{\vartwo}{u} %
\DeclareMathOperator{\varvar}{y} %
\DeclareMathOperator{\varvartwo}{v} %

\DeclareMathOperator{\varglob}{z} %
\DeclareMathOperator{\varvarglob}{w} %

\DeclareMathOperator{\varglobold}{\hat{z}} %
\DeclareMathOperator{\varvarglobold}{\hat{w}} %
\DeclareMathOperator{\varold}{\hat{x}} %
\DeclareMathOperator{\varvarold}{\hat{y}} %

\DeclareMathOperator{\Id}{Id} %

\newcommand{\cvar}{\text{CVar}} %

\DeclareMathOperator*{\argmin}{argmin}
\newcommand{\RR}{\mathbb{R}}

\renewcommand{\next}{\text{next}}

\newcommand{\tool}{\texttt{RPH}\xspace}

\begin{document}

\title{Randomized Progressive Hedging methods\\ for Multi-stage Stochastic Programming}

\author{Gilles Bareilles \and
        Yassine Laguel \and
        Dmitry Grishchenko \and
        Franck Iutzeler \and
        Jérôme Malick
        }

\ifdraft

\else
\titlerunning{Randomized Progressive Hedging methods}
\institute{G. Bareilles, Y. Laguel, D. Grishchenko, F. Iutzeler \at
              Univ. Grenoble Alpes \\
              \email{\{firstname.lastname\}@univ-grenoble-alpes.fr}           %
           \and
           J. Malick \at
              CNRS and LJK \\
              \email{jerome.malick@univ-grenoble-alpes.fr}
}
\date{ {Revised version} Received: date / Accepted: date}
\fi

\maketitle

\begin{abstract}
Progressive Hedging is a popular decomposition algorithm for solving multi-stage stochastic optimization problems. A computational bottleneck of this algorithm is that \emph{all} scenario subproblems have to be solved at each iteration. In this paper, we introduce randomized versions of the Progressive Hedging algorithm able to produce new iterates as soon as a \emph{single} scenario subproblem is solved. Building on the relation between Progressive Hedging and monotone operators, we leverage recent results on randomized fixed point methods to derive and analyze the proposed methods. Finally, we release the corresponding code as an easy-to-use Julia toolbox and report computational experiments  showing the practical interest of randomized algorithms, notably in a parallel context. Throughout the paper, we pay a special attention to presentation, stressing main ideas, avoiding extra-technicalities, in order to make the randomized methods accessible to a broad audience in the Operations Research community.
\end{abstract}

\section{Introduction}

\subsection{Context: decomposition of stochastic problems and computational limitations}

Stochastic optimization is a rich and active research domain with various applications in science and engineering ranging from telecommunication and medicine to finance; we refer to the two textbooks \cite{ruszczynski1997decomposition} and \cite{shapiro2009lectures} for theoretical foundations of this field and pointers to applications. Expressive stochastic models lead to large-dimensional optimization problems, that may be computationally challenging. In many applications, the randomness is highly structured (e.g. in multistage stochastic programming) and can be exploited by decomposition methods \cite[Chap.~3.9]{ruszczynski2003stochastic}.
The two main advantages of decomposition methods are that i) they replace a large and difficult stochastic programming problem by a collection of smaller problems; and ii) these smaller subproblems can usually be solved efficiently with standard off-the-shelf optimization software. As a result, decomposition methods provide an efficient and specialized methodology for solving large and difficult stochastic programming problems by employing readily available tools.

Progressive Hedging is a popular dual decomposition method for multistage stochastic programming. This algorithm was introduced in \cite{rockafellar1991scenarios} and can be interpreted as a fixed-point method over a splitting operator \cite{ruszczynski1997decomposition}.
Through this connection, Progressive Hedging is proved to be convergent for solving convex stochastic programs. There are also many applications to mixed-integer stochastic problems where Progressive Hedging acts as an efficient heuristic to get useful bounds; see e.g.\;\cite{watson2011progressive}. For historical perspectives, theoretical analysis, and references to applications, we refer to \cite{ruszczynski1997decomposition}.

Progressive Hedging tackles multi-stage stochastic problems by decomposing them over the scenarios and solving independently the smaller subproblems relative to one scenario. However the number of these subproblems grows exponentially with the number of stages, so that the computational bottleneck of this algorithm is that \emph{all} scenario subproblems have to be solved at each iteration.
 {As a decomposition method solving scenario subproblems independently, Progressive Hedging is an intrinsically parallel algorithm and admits direct parallel implementations for distributed computing systems (e.g.\;multiple threads in a machine, or multiple machines in a cluster).
In the homogeneous case (where all subproblems are solved with similar duration) such parallel implementations are efficient in practice and require no additional theoretical study; for early works discussing parallelization, see e.g.\;the doctoral dissertation\;\cite{somervell1998progressive} and the conference papers\;\cite{de1993computational,ryan2013toward}.
However, when the subproblems have different difficulties or the computing system is heterogeneous (e.g. with different machines or non-reliable communications between machines), the parallelization speed-up can be drastically degraded. Thus designing efficient, theoretically-grounded, variants of Progressive Hedging for heterogeneous distributed settings is still an on-going research topic (see e.g. the preprint\;\cite{eckstein2018asynchronous}).}

\subsection{Contribution: accessible, efficient, parallel Progressive Hedging variants}

In this paper, we present optimization methods based on Progressive Hedging having efficient parallel implementation and able to tackle large-scale multistage stochastic problems. Our variants are randomized algorithms solving subproblems incrementally, thus alleviating the synchronization barrier of the standard Progressive Hedging. When deployed on computing systems having multiple workers, our algorithms are able to make the most of the computational abilities, synchronously or asynchronously.

This work is based on the interaction of two complementary fields of research:
\begin{itemize}
    \item applications of Progressive Hedging in the OR community with expressive uncertainty models leading to large-scale multistage stochastic problems;
    \smallskip
    \item recent developments on randomization techniques in the optimization and monotone operators community,
    motivated by the distributive abilities of modern computing systems.
\end{itemize}
The connection between these two domains is natural, through the well-known interpretation of Progressive Hedging as a fixed point algorithm (see e.g.\;\cite{ruszczynski1997decomposition}). We also build on this connection to propose our efficient randomized Progressive Hedging algorithms. We pay a special attention to making our developments easily accessible for a broad audience in the OR and stochastic programming community: we explicitly derive the proposed methods from the textbook formulation of Progressive Hedging; we rely on well-established results to highlight fundamental ideas and to hide unnecessary technicalities. Furthermore, we take advantage of the recent distributive abilities of the Julia language \cite{Julia-2017} (using the \emph{Distributed} module) and provide an easy-to-use toolbox solving multistage stochastic programs with the proposed methods.

\section{Multistage stochastic programs: recalls and notation}\label{sec:notation}

In this section, we lay down the multistage stochastic model considered in this paper as well as our notation. We follow closely the notation of the textbook \;\cite[Chap.\;3]{ruszczynski2003stochastic}.

Stochastic programming deals with optimization problems involving uncertainty, modelled by random variable $\xi$, with the goal to find a feasible solution $x(\xi)$ that is optimal in some sense relatively to $\xi$. Considering an objective function $f$ and a risk measure $\mathcal{R}$, the generic formulation of a stochastic problem is
\begin{align}
\label{eq:stopb}
    \min_{x} ~~\mathcal{R}\big(f(x(\xi),\xi)\big).
\end{align}
For an extensive review of stochastic programming, see e.g. \cite{shapiro2009lectures}.

In the multistage setting, the uncertainty of the problem is revealed sequentially in $T$ stages. The random variable $\xi$ is split into $T-1$ chunks,  $\xi=(\xi_1,..,\xi_{T-1})$, and the problem at hand is to decide at each stage $t=1,\ldots,T$ what is the optimal action, $\var_t(\xi_{[1,t-1]})$, given the previous observations $\xi_{[1,t-1]}:=(\xi_1,\ldots,\xi_{t-1})$. The global variable of this problem thus writes
\begin{align*}
  \var(\xi) = (\var_1, \var_2(\xi_1), \ldots , \var_T(\xi_{[1,T-1]})) \in \mathbb{R}^{n_1} \!\times \dots \times \mathbb{R}^{n_T} \!= \mathbb{R}^{n}
\end{align*}
where $(n_1,\ldots,n_T)$ are the size of the decision variable at each stage and $n=\sum_{t=1}^T n_t$ is the total size of the problem.

We focus on the case where the random variable $\xi$ can take a finite number $S$ of values called \emph{scenarios} and denoted by $\xi^1,\ldots,\xi^S$. Each scenario occurs with probability $ {p_s} = \mathbb{P}[\xi = \xi^s]$ and is revealed in $T$ stages through one common start and a realization of the random variable $\xi^s = (\xi_1^s,\ldots,\xi_{T-1}^s)$. It is thus natural to represent the scenarios as the outcome of a probability tree, as illustrated in Figure\;\ref{fig:proba_tree}.
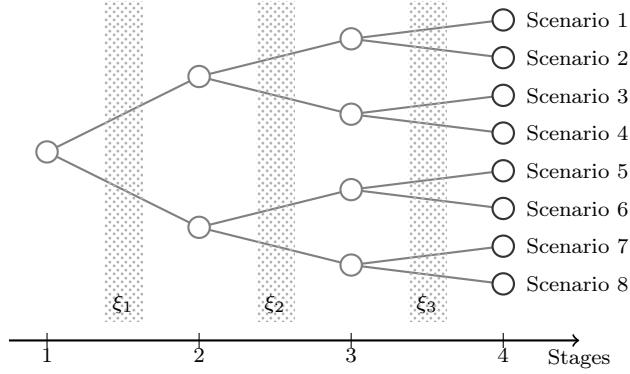
\begin{figure}[ht!]
    \centering
    \begin{tikzpicture}[]

\draw[thick, ->] (1.5,0) -- (9,0) node [below] {Stages};
    \foreach \x in {1,...,4}
    \draw (2*\x, 0.1) -- node[pos=0.5,below] {\x} (2*\x, -0.1) ;

    \draw[thick, white!50!black] (2,2.5) -- (4,1.5) ;
    \draw[thick, white!50!black] (2,2.5) -- (4,3.5) ;
    
        \draw[thick, white!50!black] (4,1.5) -- (6,1) ;
    \draw[thick, white!50!black] (4,1.5) -- (6,2)  ;
    \draw[thick, white!50!black]  (4,3.5) -- (6,3)  ;
    \draw[thick, white!50!black]  (4,3.5) -- (6,4)  ;

    \draw[thick, white!50!black]  (6,1) -- (8,0.75);
    \draw[thick, white!50!black]  (6,1) -- (8,1.25);
    
    \draw[thick, white!50!black]  (6,2) -- (8,1.75);
    \draw[thick, white!50!black]  (6,2) -- (8,2.25);
    
    \draw[thick, white!50!black]  (6,3) -- (8,2.75);
    \draw[thick, white!50!black]  (6,3) -- (8,3.25);
    
        \draw[thick, white!50!black]  (6,4) -- (8,3.75);
    \draw[thick, white!50!black]  (6,4) -- (8,4.25);

    \fill[pattern=crosshatch dots,pattern color = white!70!black] (2.75,0.25) rectangle (3.25,4.5);
    \node[above] at (3,0.25) {$\xi_1$};
    
        \fill[pattern=crosshatch dots,pattern color = white!70!black] (4.75,0.25) rectangle (5.25,4.5);
            \node[above] at (5,0.25) {$\xi_2$};
        
            \fill[pattern=crosshatch dots,pattern color = white!70!black] (6.75,0.25) rectangle (7.25,4.5);
                \node[above] at (7,0.25) {$\xi_3$};
    
    \node[draw,white!50!black,circle,fill=white,thick, scale = 1.0]  at (2,2.5) {};
    
    \node[draw,white!50!black,circle,fill=white,thick, scale = 1.0]  at (4,1.5) {};
    \node[draw,white!50!black,circle,fill=white,thick, scale = 1.0]  at (4,3.5) {};

    \node[draw,white!50!black,circle,fill=white,thick, scale = 1.0]  at (6,1.0) {};
    \node[draw,white!50!black,circle,fill=white,thick, scale = 1.0]  at (6,2.0) {};
    \node[draw,white!50!black,circle,fill=white,thick, scale = 1.0]  at (6,3.0) {};
    \node[draw,white!50!black,circle,fill=white,thick, scale = 1.0]  at (6,4.0) {};

    \node[draw,white!20!black,circle,fill=white,thick, scale = 1.0]  at (8,0.75) {};
    \node[right]  at (8.2,0.75) {Scenario $8$};
    \node[draw,white!20!black,circle,fill=white,thick, scale = 1.0]  at (8,1.25) {};
    \node[right]  at (8.2,1.25) {Scenario $7$};
    
    \node[draw,white!20!black,circle,fill=white,thick, scale = 1.0]  at (8,1.75) {};
    \node[right]  at (8.2,1.75) {Scenario $6$};
    \node[draw,white!20!black,circle,fill=white,thick, scale = 1.0]  at (8,2.25) {};
    \node[right]  at (8.2,2.25) {Scenario $5$};
    
    \node[draw,white!20!black,circle,fill=white,thick, scale = 1.0]  at (8,2.75) {};
    \node[right]  at (8.2,2.75) {Scenario $4$};
    \node[draw,white!20!black,circle,fill=white,thick, scale = 1.0]  at (8,3.25) {};
    \node[right]  at (8.2,3.25) {Scenario $3$};
    
    \node[draw,white!20!black,circle,fill=white,thick, scale = 1.0]  at (8,3.75) {};
    \node[right]  at (8.2,3.75) {Scenario $2$};
    \node[draw,white!20!black,circle,fill=white,thick, scale = 1.0]  at (8,4.25) {};
    \node[right]  at (8.2,4.25) {Scenario $1$};

\end{tikzpicture}
    \caption{Scenarios as the outcomes of  a probability tree.}
    \label{fig:proba_tree}
\end{figure}
For each scenario $s\in\{1,\ldots,S\}$, the target of multistage stochastic programming is to provide a decision $\var^s = (\var_1, \var_2(\xi^s_1), \ldots , \var_T(\xi^s_{[1,T-1]}))$, and thus the full problem variable writes
\begin{align}
\label{eq:var}
    \var = \begin{pmatrix}
    \var_1 & \var_2\left(\xi_1^1\right) & \dots & \var_{T-1}\left(\xi^1_{[1,\ldots,T-2]}\right) & \var_T\left(\xi^1_{[1,\ldots,T-1]}\right) \\
        \var_1 & \var_2\left(\xi_1^2\right) & \dots & \var_{T-1}\left(\xi^2_{[1,\ldots,T-2]}\right) & \var_T\left(\xi^2_{[1,\ldots,T-1]}\right)  \\
    \vdots& & & & \vdots\\
       \var_1 & x_2\left(\xi_1^S\right) & \dots & \var_{T-1}\left(\xi^S_{[1,\ldots,T-2]}\right) & \var_T\left(\xi^S_{[1,\ldots,T-1]}\right)
  \end{pmatrix}
  \in\mathbb{R}^{S\times n}.
\end{align}
From \eqref{eq:var}, we see that by construction of the randomness, the decision at stage\;$1$ must be the same for all the scenarios. Indeed, as no random variables have been observed, the user does not have any information about the scenarios. In the same vein, given the specificity of these random variables, an important feature of finite multistage problems is that if two scenarios $s_1$ and $s_2$ coincide up to stage\;$t-1$ (i.e. $\xi^{s_1}_{[1,t-1]} =\xi^{s_2}_{[1,t-1]}$), then the obtained decision variables must be equal up to stage\;$t$ (i.e. $(\var_1,\var_2(\xi_1^{s_1}),\ldots,\var_{t}( \xi^{s_1}_{[1,t-1]} ))=(\var_1,\var_2(\xi_1^{s_2})\ldots,\var_t(\xi^{s_2}_{[1,t-1]}))$).
These constraints are called \emph{non-anticipativity}. Geometrically these constraints define a subspace of $\mathbb{R}^{S\times n}$ that we denote by \begin{align}
    \label{eq:nonanticipativity}
    \NA \!=\!\left\{\var \in\mathbb{R}^{S\times n}\!: \forall s_1,s_2
    \left|
    \begin{array}{l}
          \var_1^{s_1} = \var_1^{s_2}  ~(t=1)\\
    ~~\text{ and } \\
    \var_t^{s_1}\!= \var_t^{s_2} \text{ if } \xi^{s_1}_{[1,t-1]}\!=   \xi^{s_2}_{[1,t-1]}~(t\geq2)\!\!
    \end{array} \right.
     \right\} {,}
\end{align}
where we denote $\var_t^s\in\mathbb{R}^{n_t}$ the decision variable for scenario $s$ at stage~$t$.
We see that the non-anticipativity constraints lead to a variable $\var$ with a block structure as depicted in Figure\;\ref{fig:block_structure_matrix}.
\begin{figure}[ht!]
\centering
\begin{tikzpicture}[]

\draw[thick, ->] (1.5,0) -- (9,0) node [below] {Stages};
    \foreach \x in {1,...,4}
    \draw (2*\x, 0.1) -- node[pos=0.5,below] {\x} (2*\x, -0.1) ;

    \draw[white!50!black, rounded corners, thick, fill = white!60!black] (1.6,0.55) rectangle (2.4,4.45);
    
    \draw[white!50!black, rounded corners, thick, fill = white!60!black] (3.6,0.55) rectangle (4.4,2.45);
    \draw[white!50!black, rounded corners, thick, fill = white!60!black] (3.6,2.55) rectangle (4.4,4.45);
    
    \draw[white!50!black, rounded corners, thick, fill = white!60!black] (5.6,0.55) rectangle (6.4,1.45);
    \draw[white!50!black, rounded corners, thick, fill = white!60!black] (5.6,1.55) rectangle (6.4,2.45);
    \draw[white!50!black, rounded corners, thick, fill = white!60!black] (5.6,2.55) rectangle (6.4,3.45);
    \draw[white!50!black, rounded corners, thick, fill = white!60!black] (5.6,3.55) rectangle (6.4,4.45);

    \draw[thick, white!80!black] (2,2.5) -- (4,1.5) ;
    \draw[thick, white!80!black] (2,2.5) -- (4,3.5) ;
    
        \draw[thick, white!80!black] (4,1.5) -- (6,1) ;
    \draw[thick, white!80!black] (4,1.5) -- (6,2)  ;
    \draw[thick, white!80!black]  (4,3.5) -- (6,3)  ;
    \draw[thick, white!80!black]  (4,3.5) -- (6,4)  ;

    \draw[thick, white!80!black]  (6,1) -- (8,0.75);
    \draw[thick, white!80!black]  (6,1) -- (8,1.25);
    
    \draw[thick, white!80!black]  (6,2) -- (8,1.75);
    \draw[thick, white!80!black]  (6,2) -- (8,2.25);
    
    \draw[thick, white!80!black]  (6,3) -- (8,2.75);
    \draw[thick, white!80!black]  (6,3) -- (8,3.25);
    
        \draw[thick, white!80!black]  (6,4) -- (8,3.75);
    \draw[thick, white!80!black]  (6,4) -- (8,4.25);

    \fill[pattern=crosshatch dots,pattern color = white!80!black] (2.75,0.25) rectangle (3.25,4.5);
    \node[above, white!70!black] at (3,0.25) {$\xi_1$};
    
        \fill[pattern=crosshatch dots,pattern color = white!80!black] (4.75,0.25) rectangle (5.25,4.5);
            \node[above, white!70!black] at (5,0.25) {$\xi_2$};
        
            \fill[pattern=crosshatch dots,pattern color = white!80!black] (6.75,0.25) rectangle (7.25,4.5);
                \node[above,  white!70!black] at (7,0.25) {$\xi_3$};
    
    \node[draw,white!80!black,circle,fill=white,thick, scale = 1.0]  at (2,2.5) {};
    
    \node[draw,white!80!black,circle,fill=white,thick, scale = 1.0]  at (4,1.5) {};
    \node[draw,white!80!black,circle,fill=white,thick, scale = 1.0]  at (4,3.5) {};

    \node[draw,white!80!black,circle,fill=white,thick, scale = 1.0]  at (6,1.0) {};
    \node[draw,white!80!black,circle,fill=white,thick, scale = 1.0]  at (6,2.0) {};
    \node[draw,white!80!black,circle,fill=white,thick, scale = 1.0]  at (6,3.0) {};
    \node[draw,white!80!black,circle,fill=white,thick, scale = 1.0]  at (6,4.0) {};

    \node[draw,white!65!black,circle,fill=white,thick, scale = 1.0]  at (8,0.75) {};
    \node[right]  at (8.2,0.75) {Scenario $8$};
    \node[draw,white!65!black,circle,fill=white,thick, scale = 1.0]  at (8,1.25) {};
    \node[right]  at (8.2,1.25) {Scenario $7$};
    
    \node[draw,white!65!black,circle,fill=white,thick, scale = 1.0]  at (8,1.75) {};
    \node[right]  at (8.2,1.75) {Scenario $6$};
    \node[draw,white!65!black,circle,fill=white,thick, scale = 1.0]  at (8,2.25) {};
    \node[right]  at (8.2,2.25) {Scenario $5$};
    
    \node[draw,white!65!black,circle,fill=white,thick, scale = 1.0]  at (8,2.75) {};
    \node[right]  at (8.2,2.75) {Scenario $4$};
    \node[draw,white!65!black,circle,fill=white,thick, scale = 1.0]  at (8,3.25) {};
    \node[right]  at (8.2,3.25) {Scenario $3$};
    
    \node[draw,white!65!black,circle,fill=white,thick, scale = 1.0]  at (8,3.75) {};
    \node[right]  at (8.2,3.75) {Scenario $2$};
    \node[draw,white!65!black,circle,fill=white,thick, scale = 1.0]  at (8,4.25) {};
    \node[right]  at (8.2,4.25) {Scenario $1$};

    \node  at (2,0.75) {\footnotesize $\var_1^8$};
    \node  at (2,1.25) {\footnotesize $\var_1^7$};
    \node  at (2,1.75) {\footnotesize $\var_1^6$};
    \node  at (2,2.25) {\footnotesize $\var_1^5$};
    \node  at (2,2.75) {\footnotesize $\var_1^4$};
    \node  at (2,3.25) {\footnotesize $\var_1^3$};
    \node  at (2,3.75) {\footnotesize $\var_1^2$};
    \node  at (2,4.25) {\footnotesize $\var_1^1$};
    
    \node  at (4,0.75) {\footnotesize $\var_2^8$};
    \node  at (4,1.25) {\footnotesize $\var_2^7$};
    \node  at (4,1.75) {\footnotesize $\var_2^6$};
    \node  at (4,2.25) {\footnotesize $\var_2^5$};
    \node  at (4,2.75) {\footnotesize $\var_2^4$};
    \node  at (4,3.25) {\footnotesize $\var_2^3$};
    \node  at (4,3.75) {\footnotesize $\var_2^2$};
    \node  at (4,4.25) {\footnotesize $\var_2^1$};
    
    \node  at (6,0.75) {\footnotesize $\var_3^8$};
    \node  at (6,1.25) {\footnotesize $\var_3^7$};
    \node  at (6,1.75) {\footnotesize $\var_3^6$};
    \node  at (6,2.25) {\footnotesize $\var_3^5$};
    \node  at (6,2.75) {\footnotesize $\var_3^4$};
    \node  at (6,3.25) {\footnotesize $\var_3^3$};
    \node  at (6,3.75) {\footnotesize $\var_3^2$};
    \node  at (6,4.25) {\footnotesize $\var_3^1$};
    
    \node  at (8,0.75) {\footnotesize $\var_4^8$};
    \node  at (8,1.25) {\footnotesize $\var_4^7$};
    \node  at (8,1.75) {\footnotesize $\var_4^6$};
    \node  at (8,2.25) {\footnotesize $\var_4^5$};
    \node  at (8,2.75) {\footnotesize $\var_4^4$};
    \node  at (8,3.25) {\footnotesize $\var_4^3$};
    \node  at (8,3.75) {\footnotesize $\var_4^2$};
    \node  at (8,4.25) {\footnotesize $\var_4^1$};

\end{tikzpicture}
\caption{Structure of the non-anticipativity constraints corresponding to the 4- {s}stage stochastic problem depicted in Fig.~\ref{fig:proba_tree}. All variables in a dark gray rectangle have to be equal.\label{fig:block_structure_matrix}}
\end{figure}
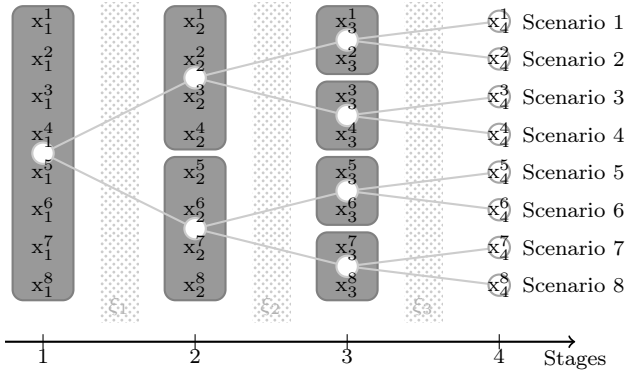

For each scenario $s\in\{1,\ldots,S\}$, let us denote by $f^s(\var^s) = f(\var^s,\xi^s)$ the cost of the decision $\var^s$. To simplify notation, we consider that possible constraints are incorporated in the cost: minimizing a function $\tilde{f}^s$ over a constraint set $\mathcal{C}^s$ is the same as minimizing $f^s = \tilde{f}^s + \iota_{\mathcal{C}^s}$ over the full space (with the indicator function $\iota_{\mathcal{C}^s}$ defined by $\iota_{\mathcal{C}^s}(x) = 0$ if $x\in\mathcal{C}^s$ and $+\infty$ elsewhere). We consider such a constrained problem in our numerical experiments in Section\;\ref{sec:num}.

We consider the ``risk-neutral case'' where $\mathcal{R}$ is the expectation of the random cost $f^s(\var^s)$ (a ``risk-averse'' case is discussed later in Remark\;\ref{sec:cvar}). In this setting, Problem~\eqref{eq:stopb} rewrites as
\begin{align}\label{eq:pb}
    \min_{ \var \in \NA} ~~\sum_{s=1}^S p_s f^s(\var^s),
\end{align}
which will be our target problem in this paper. The difficulty of this problem comes from the fact that there is an exponentially growing number of scenarios (e.g. in a binary tree,  {$S=2^{T-1}$}), all linked by the non-anticipativity constraints.

We finish presenting the set-up by formalizing our blanket assumptions on \eqref{eq:pb}.
\begin{assumption}
\label{hyp:gen}
The scenario probabilities are positive ($p_s>0$);
the  functions $f^s$ are convex, proper, and lower-semicontinuous; and there exists a solution to \eqref{eq:pb}.
\end{assumption}

\begin{assumption}
\label{hyp:gen2}
The subspace\;$\NA$ defined in\;\eqref{eq:nonanticipativity} intersects the relative interior of the domain of the objective function.
\end{assumption}

The convexity in Assumption\;\ref{hyp:gen} is used for the convergence analysis (see e.g.\;\cite[Chap.\;3]{ruszczynski2003stochastic}).
The technical assumption\;\ref{hyp:gen2}
is the standard non-degeneracy assumption in the multi-stage stochastic programming (see e.g.\;(9.17) in\;\cite[Chap.\;3]{ruszczynski2003stochastic}) which enables the splitting between scenarios and constraints.

\begin{remark}[Risk-averse variant]
\label{sec:cvar}
Though we consider in~\eqref{eq:pb} a risk-neutral model, this formulation naturally extends to risk-averse models, for which ``worst'' scenarios are particularly important to take into account. A popular risk-averse measure is the so-called Conditional Value at risk or CVar (see e.g.\;\cite{rockafellar2018superquantile}).
Following the idea of \cite{rockafellar2018solving}, the risk-averse problem
$$ \min_{\var \in \NA} ~\cvar_p(f(\var(\xi), \xi))$$
can be cast in the same form as \eqref{eq:pb}. %
\qed
\end{remark}

\section{Progressive Hedging: algorithm and sequential randomization}\label{sec:algorithms}

This section presents %
an efficient randomization of Progressive Hedging for solving the multi-stage stochastic problem~\eqref{eq:pb}. We start with recalling the usual Progressive Hedging algorithm and discussing its practical implementation. Then, we propose a \emph{sequential} randomized variant, %
which is a single-thread method, as the standard Progressive Hedging, but with cheap iterations.

This variant, as well as the other upcoming methods of the next section, is based on the operator view of Progressive Hedging (see e.g.\;the textbook\;\cite[Chap.~3.9]{ruszczynski2003stochastic}). More precisely, Progressive Hedging corresponds to the Douglas-Rachford splitting on the subgradient of the dual problem, much like the Alternating Direction Method of Multipliers\;(ADMM); see\;\cite{lions1979splitting}. We refer to \cite{eckstein1992douglas} for a formal link between Douglas-Rachford and ADMM, and  to\;\cite[Chap.~3.9]{ruszczynski2003stochastic} for a formal link between Douglas-Rachford and Progressive Hedging. However, no knowledge on fixed-point theory is required to read this section; we postpone the derivation of the algorithms and the proofs of the theorems in Appendix\;\ref{sec:OPPH}.

\subsection{Progressive Hedging}\label{sec:standard}

Progressive Hedging is a popular decomposition method for solving \eqref{eq:pb} by decoupling the objective (separable among the scenarios) and the constraints (linking scenarios). The method alternates between two steps: i) solving $S$ subproblems (one for each scenario, corresponding to the minimizing $f^s$ plus a quadratic function) independently; ii) projecting onto the non-anticipativity constraint. In order to properly define this second step, it is convenient to define the \emph{bundle}  $\mathcal{B}^s_t$ of the scenarios that are indistinguishable from scenario $s$ at time $t$, i.e.
\begin{align*}
    \mathcal{B}^s_t &= \left\{ \sigma\in\{1,\ldots,S\} :   \xi^{s}_{[1,t-1]} =   \xi^{\sigma}_{[1,t-1]}   \right\}  \\
    \nonumber  &=  \left\{ \sigma\in\{1,\ldots,S\} :  \var\in\NA \Rightarrow \var^{s}_{t} =   \var^{\sigma}_{t}   \right\}  ~~~~ \text{ (see\;\eqref{eq:nonanticipativity}) } .
\end{align*}
Projecting onto the non-anticipativity constraints decomposes by stage and scenario, as an average over the corresponding bundle weighted by the scenarios probabilities.

\begin{Algorithm}
\begin{align*}
& \text{\bfseries Initialize: } \var^0 \in \NA, \vartwo^0 \in \NA^\perp, \mu>0  \\
& \text{\bfseries For } k = 0,1,\ldots ~ \text{\bfseries do:} \\
   &  \left\{
        \begin{array}{ll}
            \varvar^{k+1,s} = \argmin_{y\in\RR^n}\left\{ f^s(y) + \frac{1}{2\mu} \left\|y-\var^{k,s} + \mu \vartwo^{k,s} \right\|^2  \right\} \text{ for all } s=1,\ldots,S & \\
            \var_t^{k+1,s} = \frac{1}{\sum_{\sigma\in\mathcal{B}^s_t} p_\sigma } \sum_{\sigma\in\mathcal{B}^s_t} p_\sigma \varvar_t^{k+1,\sigma}  \text{ for all } s=1,\ldots,S \text{ and } t=1,\ldots,T  \\
            \vartwo^{k+1} = \vartwo^{k} + \frac{1}{\mu} (\varvar^{k+1}-\var^{k+1}) &
        \end{array}
    \right.  \\
    & \text{\bfseries Return: } \var^k
\end{align*}
\caption{Progressive Hedging \label{alg:PH}}
\end{Algorithm}

Algorithm\;\ref{alg:PH} presents the Progressive Hedging algorithm; its derivation from the reformulation of \eqref{eq:pb} as fixed-point problem is recalled in Appendix~\ref{sec:OPPH}. This appendix also explains how the convergence of the algorithm can be obtained as an application of existing results for fixed-point algorithms. Here we only formalize the convergence result, and discuss further some implementation details.

\begin{theorem}\label{th:ph}
Consider the multistage problem \eqref{eq:pb} verifying Assumption\;\ref{hyp:gen}\;and\;\ref{hyp:gen2}. Then, the sequence $(\var^k)$ generated by Algorithm\;\ref{alg:PH}
is feasible ($\var^k\in\NA$ for all $k$) and
converges to an optimal solution of\;\eqref{eq:pb}.
\end{theorem}

The costly operation in Algorithm\;\ref{alg:PH} is the update of the variable $\varvar$ which consists in a ``proximal'' operation for \emph{all} scenarios. In general, there is no closed form expression for this operation, and thus it has to be obtained by a nonlinear optimization solver\footnote{In our toolbox, we solve these problems with IPOPT  {\cite{wachter2006implementation}}, an open source software package for nonlinear optimization.}. For instance, when the scenario costs $(f^s)$ are (constrained) linear or quadratic functions, this operation amounts to solving $S$ (constrained) quadratic programs.

The update of the variable $\var$ consists in a projection onto the non-anticipativity constraints $\NA$. Though it is rather cheap to compute, it involves the variables $(\varvar^s)$ of all the scenarios. Thus, the update of $\varvar$ has to be completely executed before updating\;$\var$, resulting in a potential computational bottleneck. Our upcoming randomized variant is aimed at alleviating this practical drawback.

Finally, concerning the initialization of the method, $\vartwo^0 \in \NA^\perp$ is primordial for convergence, and $\vartwo^0 = 0$ is a safe choice. The hyperparameter $\mu>0$ (also present in ADMM) controls the relative decrease of the primal and dual error; for the specific structure of Problem~\eqref{eq:pb}, $\mu=1$ seem to be an acceptable choice in most situations.

 {
\begin{remark}[About leaf nodes] \label{rem:leaf}
The \emph{leaf nodes} correspond to the variables $\var_T^s$ i.e. the decisions at the last stage $T$ for all the scenarios (see Fig.~\ref{fig:block_structure_matrix}). These variables are particular in the optimization problem since they are not linked by the non-anticipativity constraints. This implies that the scenario bundle for scenario $s$ at the last stage $T$ is reduced to the singleton $\mathcal{B}_T
^s = \{s\}$. Hence, the weighted average in the update of\;$\var$ in Algorithm~\ref{alg:PH} reduces to $\var_T^{k+1,s} = \varvar_T^{k+1,s}$, %
and as a consequence
the dual variable %
stays unchanged $\vartwo^k_T = \vartwo^0_T$. %
This modification was present in the original Progressive Hedging algorithm by Rockafellar and Wets \cite{rockafellar1991scenarios} since it allows to reduce the storage cost by getting rid of the dual variable corresponding to leaf nodes. This does not hold anymore for randomized versions that we present in this paper. So we choose not to display this modification for simplicity and better compliance with the textbook \cite{ruszczynski2003stochastic}.
\qed
\end{remark}
}

\subsection{Randomized Progressive Hedging}\label{sec:rand}

Using the interpretation of Progressive Hedging as a fixed-point algorithm (detailed in appendix\;\ref{sec:OPPH}) and results on randomized ``coordinate descent'' fixed-point methods (see e.g.\;\cite{iutzeler2013asynchronous}), we obtain a randomized version of Progressive Hedging. This randomized method consists in updating only a subset of the coordinates at each iteration, corresponding to only \emph{one} scenario, and leaving the other unchanged.  {By doing so, each iterations is much less demanding computationally than one of Progressive Hedging (roughly $S$ time quicker) since it involves the resolution of one sub-problem compared to $S$. However, as commonly observed with randomized optimization methods, Algorithm~\ref{alg:RPH} will take more iterations to converge but usually less than $S$ times more due to the progressive improvement brought by each scenario information. Thus, the Randomized Progressive Hedging should in general outperform the Progressive Hedging computationally, with the other advantage that many more iterations are produced per time which can be very useful in practice.} Deriving such a  method requires a special care as the operator links the variables with each other; these derivations are provided in Appendix~\ref{apx:rph}.

\begin{Algorithm}
\begin{align*}
& \text{\bfseries Initialize: } \varglob^0 = \var^0  \in \NA, \mu>0 \\
& \text{\bfseries For } k = 0,1,\ldots ~ \text{\bfseries do:} \\
 &   \left\{
        \begin{array}{ll}
           \text{Draw a scenario } \selected \in \{1,\ldots,S\} \text{ with probability } \mathbb{P}[\selected = s] = q_s  & \\[1.5ex]
            \displaystyle\var_t^{k+1,\selected} = \frac{1}{\sum_{\sigma\in\mathcal{B}^{\selected}_t} p_\sigma } \sum_{\sigma\in\mathcal{B}^{\selected}_t} p_\sigma \varglob_t^{k,\sigma}  \text{ for all }  t=1,\ldots,T  &
            {\scriptstyle \text{projection only on}} \\[-0.6cm]
            &
            {\scriptstyle \text{constraints involving } \selected } \\[0.3cm]
            \varvar^{k+1,\selected} = \argmin_{y\in\RR^n}\left\{ f^{\selected}(y) + \frac{1}{2\mu} \left\|y- 2\var^{k+1,\selected} + \varglob^{k,\selected} \right\|^2  \right\}& {\scriptstyle \text{ optimization sub-problem}} \\[-0.3cm]
            & {\scriptstyle \text{ only concerning scenario } \selected } \\[0.3cm]
            \left| \begin{array}{l}
                \varglob^{k+1,\selected} = \varglob^{k,\selected} + \varvar^{k+1,\selected} - \var^{k+1,\selected}  \\
            \varglob^{k+1,s} =  \varglob^{k,s} \text{ for all } s\neq \selected
            \end{array}
             \right.
            &
        \end{array}
    \right. \\
    & \text{\bfseries Return: } \tilde{\var}^{k+1} = \frac{1}{\sum_{\sigma\in\mathcal{B}^s_t} p_\sigma } \sum_{\sigma\in\mathcal{B}^s_t} p_\sigma \varglob_t^{k+1,\sigma}  \text{ for all } s=1,\ldots,S \text{ and } t=1,\ldots,T
\end{align*}
\caption{Randomized Progressive Hedging \label{alg:RPH}}
\end{Algorithm}

At iteration $k$, our Randomized Progressive Hedging (Algorithm~\ref{alg:RPH}) samples one scenario $\selected$ (randomly among all with probabilities $(q_s)$) and then alternates between
the projection over the non-anticipativity constraints $\NA$ \emph{associated with  $\selected$} (the full projection on $\NA$ is not necessary\footnote{The full projection can be performed anyway but the variables that are not associated with  $\selected$ will not be taken into account by the algorithm anyhow.}) and the ``proximal'' operation over the selected scenario\;$\selected$, together with an update of the main variable $z$. Since a single scenario is involved in the iteration, this algorithm is naturally adapted to single-thread implementations and its incremental nature makes it computationally more efficient than the Progressive Hedging, to almost no additional implementation complications.

Finally, notice that since only the partial projection on the non-anticipativity constraints related to this scenario is needed to perform an iteration, the sequence $(\var^k)$, although converging to the sought solution, does not verify the non-anticipativity constraints. That is why the output of the algorithm has to be eventually projected onto the (full) non-anticipativity constraints (which introduces variable $\tilde \var^k$).

The convergence of this randomized version is formalized by  Theorem\;\ref{th:rph} and proved in Appendix~\ref{apx:rph}.

\begin{theorem}
\label{th:rph}
Consider a multistage problem \eqref{eq:pb} verifying Assumptions\;\ref{hyp:gen}\;and\;\ref{hyp:gen2}.  Then, the sequence $(\tilde{\var}^k)$ generated by Algorithm\;\ref{alg:RPH} is feasible (\,$\tilde{\var}^k\!\in\!\NA$~a.s.\;for\;all\;$k$) and converges almost surely to a solution of~\eqref{eq:pb}.
\end{theorem}

In practice, the initialization and parameters are similar to those of Progressive Hedging to the exception of the probabilities $(q_s)$. Two natural choices come to mind:
\begin{itemize}
    \item \emph{uniform sampling}: taking the same probability $q_s = 1/S$ for all scenarios;
      \item \emph{p sampling}: taking $q_s = p_s$ and thus sampling more the scenarios with a greater weight in the objective.
\end{itemize}
Finally, in terms of implementation, this algorithm is by nature \emph{sequential} in the sense that one scenario is sampled, treated, and then incorporated in the master variable. Thus, it suits well single thread setups but is not directly able to benefit from multiple workers. Such an extension is the goal of the next section.

\section{Parallel variants of Progressive Hedging}

In this section, we discuss the deployment of (variants of) Progressive Hedging algorithms on parallel computing systems. No particular knowledge about distributed computing is required. We consider a generic master-worker framework where $M$ workers collaboratively solve \eqref{eq:pb} under the orchestration of a master. This setting encompasses diverse practical situations such as multiple threads in a machine or multiple machines in a computing cluster (workers can then be threads, machines, agents, oracles, etc.).  {Our aim is to provide parallel methods that speed-up the resolution of medium-to-large multi-stage stochastic programs by using a distributed computing system. Fully scalable implementations are problem/system dependent; instead, we} take a higher level of abstraction and consider that a worker is a computing procedure  {that is able to solve any given subproblem}. In practice, the algorithms implemented in our toolbox do not need to know the computing system, as they automatically adapt the underlying computing system, thanks to parallelization abilities of the Julia language.

\subsection{Parallel Progressive Hedging}\label{sec:par}

The randomized method presented in Section~\ref{sec:rand} is based on the sampling of one scenario per iteration. Using the same reasoning, it is possible to produce an algorithm sampling $M\leq S$ scenarios per iteration. These $M$ scenarios can then be treated in parallel by $M$ workers and then sent to the master for incorporation in the master variable. This algorithm, completely equivalent to the Randomized Progressive Hedging (Algorithm~\ref{alg:RPH}) can be formulated in a master-worker setup as follows.

\begin{Algorithm}[!h]
\begin{center}
{\footnotesize
\begin{worker}{Master}
\vspace*{-3mm}
\begin{align*}
& \text{\bfseries Initialize: } \var^0=\varglob^0 \in \NA, \mu>0  \\
& \text{\bfseries For } k = 0,1,\ldots ~ \text{\bfseries do:} \\
&    \left\{
        \begin{array}{ll}
         \text{Draw $M$ scenarios } (s[1],..,s[M]) \in \{1,\ldots,S\}^M \text{ with probability } \mathbb{P}[s[i] = s] = q_s  &   \\
         \var^{k+1,s}_{t} = \frac{1}{\sum_{\sigma\in\mathcal{B}^{s}} p_\sigma } \sum_{\sigma\in\mathcal{B}^{s}_t} p_\sigma \varglob_t^{k,\sigma}  \text{ for all }  t=1,\ldots,T   \text{ and } s \in (s[1],..,s[M]) &  \\
         \text{\underline{Send} a scenario/point pair } (s[i],  {2 \var^{k+1,s[i]} -  \varglob^{k,s[i]} }  ) \text{ to \emph{each} worker } i=1,..,M  & \\
         \text{\underline{Receive} } \varvar^{s[i]}  \text{ from \emph{all} workers } i=1,..,M  & \\
                    \left| \begin{array}{l}
                \varglob^{k+1,s} = \varglob^{k,s} +  \varvar^{s} -   {\var^{k+1,s} } \text{ for all } s \in (s[1],..,s[M])  \\
            \varglob^{k+1,s} =  \varglob^{k,s} \text{ for all } s \notin (s[1],..,s[M])
            \end{array}
             \right. & \\
        \end{array}
    \right. \\
    & \text{\bfseries Return: }  \tilde{\var}^{k+1} = \frac{1}{\sum_{\sigma\in\mathcal{B}^s_t} p_\sigma } \sum_{\sigma\in\mathcal{B}^s_t} p_\sigma  {\varglob}_t^{k+1,\sigma}  \text{ for all } s=1,\ldots,S \text{ and } t=1,\ldots,T
\end{align*}
\end{worker}
\vspace*{-3mm}
\begin{worker}{Worker $i$}
\vspace*{-3mm}
\begin{align*}
& \text{\bfseries As soon as a scenario/point pair is received:} \\
&    \left\{
        \begin{array}{ll}
        \text{\underline{Receive} scenario/point pair } (s[i],\varvartwo[i])   & \\
            \varvar^{s[i]} = \argmin_{y\in\RR^n}\left\{ f^{s[i]}(y) + \frac{1}{2\mu} \left\|y- \varvartwo[i] \right\|^2  \right\} &  \\
            \text{\underline{Send} } \varvar^{s[i]}  \text{ to the Master }  & \\
        \end{array}
    \right.
\end{align*}
\end{worker}}
\end{center}
\caption{Parallel Randomized Progressive Hedging\label{alg:RPHPar}}
\end{Algorithm}

This algorithm presents a simple, yet rather efficient, parallel method to solve multistage stochastic problems based on Progressive Hedging. When the difficulty of the subproblems is highly variable (due to different sizes, data, or initialization), the Progressive Hedging has to wait for the slowest subproblem to be solved, in order to complete an iteration. This is not the case anymore for the parallel randomized variant. However, if the \emph{computing system} is heterogeneous, the parallel version still has to wait for the slowest worker, and thus workers may eventually have idle times. This drawback occurring for heterogeneous setups will  be alleviated in the next section by our asynchronous variant.

\subsection{Asynchronous Randomized Progressive Hedging}\label{sec:async}

In a parallel computing framework, the Parallel Randomized Progressive Hedging of the previous section can be further extended to generate asynchronous iterations (built on a slightly different randomized fixed-point method\;\cite{peng2016arock}, as detailed in Appendix~\ref{apx:arph}).

The resulting asynchronous Progressive Hedging (Algorithm\;\ref{alg:ARPH}) consists of the same steps per iteration as Algorithm\;\ref{alg:RPH}, but these steps are performed asynchronously by several workers in parallel. In this case, each of the workers asynchronously receives a global variable, computes an update associated with one randomly drawn scenario, then incorporates it to the master variable\footnote{We assume consistent writing, i.e. reading and writing do not clash with each other, extensions to inconsistent reads is discussed in \cite[Sec.\;1.2]{peng2016arock}}.

 {
\begin{Algorithm}[!h]
\begin{center}
{\footnotesize
\begin{worker}{Master}
\vspace*{-3mm}
\begin{align*}
& \text{\bfseries Initialize: } \var^0=\varglob^0 \in \NA, \mu>0, k=0,  \\
& \hspace*{1.55cm} \varold[j]=\var^{0,j} \text{ and } s[j]=j \text{ for every worker } i   \\
& \text{\underline{Send} the scenario/point pair } (s[j],\var[j]) \text{ to every worker } j  \\
& \text{\bfseries As soon as a worker finishes its computation:} \\
&    \left\{
        \begin{array}{ll}
         \text{\underline{Receive} } \varvar^{s[i]}  \text{ from an worker, say } i & \\
                    \left| \begin{array}{l}
                \varglob^{k+1,s[i]} = \varglob^{k,s[i]} + \frac{2 \eta^k}{S q_{s[i]}} \left(  \varvar^{s[i]} -  \varold[i] \right)  \\
            \varglob^{k+1,s\hphantom{[i]}} =  \varglob^{k,s} \text{ for all } s\neq s[i]
            \end{array}
             \right. & \\
           \text{Draw a new scenario for }i: s[i] \in \{1,\ldots,S\} \text{ with probability } \mathbb{P}[s[i] = s] = q_s  &   \\
        \varold[i] = \frac{1}{\sum_{\sigma\in\mathcal{B}^{s[i]}} p_\sigma } \sum_{\sigma\in\mathcal{B}^{s[i]}_t} p_\sigma \varglob_t^{k+1,\sigma}  \text{ for all }  t=1,\ldots,T  &  \\
             \text{\underline{Send} the scenario/point pair } (s[i],2 \varold[i] -  \varglob^{k+1,s[i]}) \text{ to worker } i  & \\
        k \leftarrow k+1 &
        \end{array}
    \right. \\
    & \text{\bfseries Return: }  \tilde{\var}^{k+1} = \frac{1}{\sum_{\sigma\in\mathcal{B}^s_t} p_\sigma } \sum_{\sigma\in\mathcal{B}^s_t} p_\sigma  {\varglob}_t^{k+1,\sigma}  \text{ for all } s=1,\ldots,S \text{ and } t=1,\ldots,T
\end{align*}
\end{worker}
\vspace*{-3mm}
\begin{worker}{Worker $i$}
\vspace*{-3mm}
\begin{align*}
& \text{\bfseries As soon as a scenario/point pair is received:} \\
&    \left\{
        \begin{array}{ll}
        \text{\underline{Receive} scenario/point pair } (s[i],\varvartwo[i])   & \\
            \varvar^{s[i]} = \argmin_{y\in\RR^n}\left\{ f^{s[i]}(y) + \frac{1}{2\mu} \left\|y- \varvartwo[i] \right\|^2  \right\} &  \\
            \text{\underline{Send} } \varvar^{s[i]}  \text{ to the Master }  & \\
        \end{array}
    \right.
\end{align*}
\end{worker}}
\end{center}
\caption{Asynchronous Randomized Progressive Hedging \label{alg:ARPH}}
\end{Algorithm}
}

Multiple updates may have occurred between the time of reading and updating.
 {
We denote by $\varold[i]$ (without any time index to avoid confusion) the value of $\var^{k,s}$ lastly used for feeding worker $i$. When the master performs an update from the information of worker $i$, $\varold[i]=\var^{k-d_k,s[i]}$ where $s[i]$ is the scenario treated by worker $i$ for that update and $d^k$ if the number of updates between $k$ and the last time worker $i$ performed an update\footnote{In Appendix~\ref{apx:arph}, following \cite{peng2016arock}, we denote by $\varold^k=\var^{k-d_k}$ if worker $i$ started its update at time $k-d_k$.}.
}
We assume here that this delay is uniformly bounded; this is a reasonable assumption for multi-core machines and computing clusters. This assumption allows to use the convergence analysis of\;\cite{peng2016arock} to establish the following convergence result. The proof of this result is given in Appendix~\ref{apx:arph}.  {The intuition behind the result is to use the maximal delay to take cautious stepsizes $\eta_k$, guaranteeing convergence of asynchronous updates.}

\begin{theorem}\label{th:arph}
Consider a multistage problem \eqref{eq:pb} verifying Assumption {s}\;\ref{hyp:gen}\;and\;\ref{hyp:gen2}.
We assume furthermore that the delays are bounded:
$d^k\leq \tau <\infty $ for all $k$.
If we take the stepsize $\eta^k$ as follows for some fixed  $0<c<1$
\begin{equation}\label{eq:bounddelay}
0< \eta_{\min} \leq \eta^k ~\leq~  \frac{c S q_{\min}}{2\tau \sqrt{q_{\min}} +1 }
\qquad \text{with $q_{\min} = \min_s q_s$},
\end{equation}
then, the sequence $(\tilde{\var}^k)$ generated by Algorithm\;\ref{alg:ARPH} is feasible ($\tilde{\var}^k\!\in\!\NA$~a.s.\;for all\;$k$) and converges almost surely to a random variable supported by the solution set of\;\eqref{eq:pb}.
\end{theorem}

 {\begin{remark}[Extensions]
For sake of clarity, we reduce Algorithm\;\ref{alg:ARPH} to its simplest formulation with essential ingredients for asynchronous computation with guaranteed convergence. Several extensions and heuristics could be added; among them:
\begin{itemize}
    \item tuned $\eta_k$ (we test the simple strategy $\eta_k=1$ in our numerical experiments),
    \item adaptive $\mu$ (scenario or iteration-wise),
    \item sending multiple scenario/point pairs (instead of only one) to the updating worker.\qed
\end{itemize}
\end{remark}}

 {\begin{remark}[Comparison with the other existing asynchronous variant]

\noindent
The preprint\;\cite{eckstein2018asynchronous} proposes another asynchronous variant of Progressive Hedging. This algorithm obviously shares common points with Algorithm\;\ref{alg:ARPH} but has fundamental differences.
The most striking one lies in the primal-dual update: at each iteration we update the primal-dual variable $\varglob^k$ only for the current scenario $s[i]$ while
the asynchronous Progressive Hedging of\;\cite{eckstein2018asynchronous} updates the full primal and dual variables.
This comes from the fact that our algorithm is based on the asynchronous coordinate-descent method for operators of\;\cite{peng2016arock} while
\cite{eckstein2018asynchronous} is based on the asynchronous splitting method of\;\cite{eckstein2017simplified}.
Another practical difference is that we only use a partial projection related to the drawn scenario.\qed
\end{remark}}

\section{The \tool toolbox}\label{sec:toolbox}

We release an open-source toolbox for modeling and solving multi-stage stochastic problems with the proposed Progressive Hedging variants. The toolbox is named \tool (for Randomized Progressive Hedging) and is implemented on top of \texttt{JuMP} \cite{DunningHuchetteLubin2017} the popular framework for mathematical optimization, embedded in \texttt{Julia} language \cite{Julia-2017}. The source code, online documentation, and an interactive demonstration are available on the GitHub page of the project:
\begin{center}
    \url{https://github.com/yassine-laguel/RandomizedProgressiveHedging.jl}.
\end{center}

 {The toolbox \tool seems to implement the first publicly-available and theoretically-grounded variant of progressive hedging with randomized/asynchronous calls to the scenario subproblems. Related implementations include the theoretically-grounded one of \cite{eckstein2018asynchronous}
and the asynchronous heuristic of \texttt{ProgressiveHedgingSolvers.jl}
publicly-available via the modelling framework \texttt{StochasticPrograms.jl} \cite{spjl}.}

In this section, we only describe the basic usage of \tool; for more details, we refer to the Appendix\;\ref{sec:rph_details} and the online material. Notably, we defer the problem modeling (for direct testing, we provide functions directly building toy problems such as \texttt{build\_simpleexample}). Once a problem is instantiated, its resolution can be launched directly with the sequential Progressive Hedging or Randomized Progressive Hedging methods (Algorithms\;\ref{alg:PH} and\;\ref{alg:RPH}):
\begin{lstlisting}
using RPH

pb = build_simpleexample() ## Generation of a toy problem

y_PH = solve_progressiveHedging(pb) ## Solving with Progressive Hedging
println("Progressive Hedging output is: ", y_PH)

y_RPH = solve_randomized_sync(pb)  ## Solving with Randomized Progressive Hedging
println("Randomized Progressive Hedging output is:  ", y_RPH)
\end{lstlisting}

In a parallel environment, one can use the \texttt{Distributed} module of \texttt{Julia}\footnote{A detailed explanation of \texttt{Julia}'s parallelism is available at \texttt{Julia} documentation: \url{https://docs.Julialang.org/en/v1/manual/parallel-computing/}. By default, the created workers are on the same machine but can easily be put on a distant machine through an SSH channel.} to setup some number of workers.  This can be done either at launch time (e.g. as \texttt{Julia -p 8}) or within the \texttt{Julia} process (with commands \texttt{addprocs()}, \texttt{rmprocs()} and \texttt{procs()}). Once this is set, the proposed parallel algorithms (Algorithm\;\ref{alg:RPHPar} and Algorithm\;\ref{alg:ARPH}) automatically use all available workers:

\begin{lstlisting}
using Distributed
addprocs(7); length(procs()) # Gives one master + 7 workers

y_par = solve_randomized_par(pb)   ## Solving with Parallel Randomized Progressive Hedging
println("Parallel Randomized Progressive Hedging output is: ", y_par)

y_async = solve_randomized_async(pb)  ## Solving with Asynchronous Randomized Progressive Hedging
println("Asynchronous solve output is: ", y_async)
\end{lstlisting}
The provided methods in \tool rely on  {three} criteria for stopping:
\begin{itemize}
    \item maximal computing time (default: one hour),
    \item maximal number of scenario subproblems solved (default:  {$10^6$}),
    \item  {residual norm (norm of differences of iterates) %
     inferior to the mixed absolute/relative threshold $\varepsilon_{abs}+\varepsilon_{rel}\|z_k\|$, where $z_k$ is the current iterate of the algorithm (default: $\varepsilon_{abs}=10^{-8}$, $\varepsilon_{rel}=10^{-4}$).}
\end{itemize}

\section{Numerical illustrations}\label{sec:num}

This section presents numerical results obtained with our toolbox \tool on multistage stochastic problems. We illustrate the behavior of our methods on a small hydro-thermal scheduling problem. A complete experimental study on real-life problems or modern parallel computing systems is beyond the scope of this work. We release our toolbox to allow reproducibility of our results and to spark further research on these randomized method {s}.

\subsection{A multistage stochastic problem inspired from energy optimization}\label{sec:toy}

We consider a simple convex problem modeling a problem of hydro thermal scheduling; it follows\;\cite{pereira1991multi} and the FAST toolbox\footnote{\url{https://stanford.edu/~lcambier/cgi-bin/fast/index.php}}.

Assume that an energy company wishes to deliver electricity to its clients either produced by several dams or bought externally. The dams produce cheaper energy but can only store a limited amount of water. The randomness of the problem comes from the amount of rain filling the dams at each stage. Mathematically, at each stage $t\in\{1,\ldots,T\}$, each dam $b\in\{1,\ldots,B\}$ has a quantity $q_t^b\in\mathbb{R}_+$ of water. For each stage $t$, the company has to decide: i) for each dam $b$ the quantity of water to convert to electricity $y_t^b\in\mathbb{R}_+$; and ii) the quantity of electricity to buy externally $e_t\in\mathbb{R}_+$. The decision variable at stage $t$ thus writes $\var_t = (q_t,y_t,e_t)\in\mathbb{R}^B_+\times\mathbb{R}^B_+\times\mathbb{R}_+$.

At stage $t$, the random variable $\xi^t$ represents the amount of water that arrived at each of the dams since stage $t-1$. Out of simplicity, $\xi_t$ is equal to $r_{dry}$  with probability $p_{dry}$ or $r_{wet}$ with probability $1-p_{dry}$. This defines a binary scenario tree (see Fig.~\ref{fig:proba_tree}) leading to  {$2^{T-1}$} scenarios.

For a scenario $s$, i.e. a realization of the sequence of water arrivals $(\xi_1^s,\ldots,\xi_{T-1}^s)$, the objective function writes as  {the sum $f^s = \tilde{f}^s + \iota_{\mathcal{C}^s}$ of the energy generation cost
\[
    \tilde{f}^s(\var) = \sum_{t=1}^T c_{H, t}^\top y_t + c_E  e_t
\]
and the indicator function of constraints
\begin{align*}
\mathcal{C}^s =   \left\{
    \begin{array}{lll}
        \sum_{b=1}^B y_t^b + e_t \geq D &  \text{ for all } t & {\scriptstyle \text{(demand is met at each stage)}}   \\
        q_t^b = q_{t-1}^b  - y_{t}^b + \xi_t^ {s} &\text{ for all } t \geq 2,b & {\scriptstyle \text{(evolution of the amount of water)}}   \\
        q_1^b  = W_1^b - y_1^b & \text{ for all } b & {\scriptstyle \text{(init. amount of water per dam)}}   \\
                q_t^b \leq W^b & \text{ for all } t,b & {\scriptstyle \text{(max. amount of water per dam)}}
    \end{array}
    \right. .
\end{align*}}
For a given scenario, minimizing this objective function amounts to solving a quadratic optimization problem. The variables and constants are summarized in Table~\ref{tab:hydro}.

\begin{table}[ht!]
    \centering

        \begin{tabular}{|p{0.02\textwidth}|p{0.03\textwidth}|p{0.05\textwidth}|p{0.75\textwidth}|}
    \hline
    \parbox[t]{2mm}{\multirow{3}{*}{\rotatebox[origin=c]{90}{M.S.P.}}} & $T$ & $ \mathbb{N}$ &  number of stages \\
    &$S$ & $\mathbb{N}$ & number of scenarios\\
    & $\xi$ & $\mathbb{R}_+^{T-1}$ & amount of water brought by the rain since the previous stage\\
          \hline
    \end{tabular}

        \begin{tabular}{|p{0.02\textwidth}|p{0.03\textwidth}|p{0.05\textwidth}|p{0.75\textwidth}|}
    \hline
    \parbox[t]{2mm}{\multirow{6}{*}{\rotatebox[origin=c]{90}{Constants}}} & $B$ & $ \mathbb{N}$ &  number of dams \\
    &$c_{H}$ & $\mathbb{R}^{B T}_+$ & vector of electricity production costs at the dams\\
    & $c_E$ & $\mathbb{R}_+$ & cost of buying external electricity\\
     &   $D$ & $ \mathbb{R}_+$ & electricity demand to satisfy at each stage \\
&    $W$ & $ \mathbb{R}_+^B$ & maximal amount of water at the dams \\
  &  $W_1$ & $ \mathbb{R}_+^B$ & initial amount of water available at the dams\\
          \hline
    \end{tabular}

        \begin{tabular}{|p{0.02\textwidth}|p{0.03\textwidth}|p{0.05\textwidth}|p{0.75\textwidth}|}
    \hline
    \parbox[t]{2mm}{\multirow{4}{*}{\rotatebox[origin=c]{90}{Variables}}} & $q$ & $ \mathbb{R}_+^{B T}$ &  quantity of water at the dams for each stage (directly depends on $y$ and $\xi$) \\
    &$y$ & $ \mathbb{R}_+^{B T}$ & amount of water to transform into electricity at the dams at each stage \\
    & $e$ & $\mathbb{R}_+^{T}$ & amount of electricity to purchase at each stage \\
    & $\var$ & $ \mathbb{R}_+^{n}$ & $\var = (q,y,e)$ and $n=(2B+1)T$ \\
          \hline
    \end{tabular}

    \caption{Variables and constants for the hydroelectric problem}
    \label{tab:hydro}
\end{table}

For our computational illustration, we generate randomly one instance of this problem, with $B=20$ dams,  {$T=6$ stages making $S = 2^5 =  {32}$ scenarios}. The $ {32}$ quadratic subproblems (one associated to each scenario) are solved by the interior point solver of \texttt{Mosek}, with default parameters\footnote{In particular, the (primal) feasibility tolerance is $10^{-8}$, and therefore this is the target level of tolerance for the experiments.}. We also use the solver \texttt{Mosek} to compute the optimal solution with high precision in order to plot the suboptimality of the iterates generated along the run of the algorithms.  {Since the problem is not big, this computation is quick, of the order of a second.} %

\subsection{Numerical Results}\label{sec:res}

Our illustrative experiments compare %
the behavior of the different variants of Progressive Hedging implemented in \tool. We make two experiments to illustrate the interests of randomization and parallelization for Progressive Hedging: on a sequential/single-thread set-up and on a parallel setup.

We run our experiments on a laptop with an 8-core processor (Intel(R) Core(TM) i7-10510U CPU @ 1.80GHz). For the parallel computation, one core plays the role of master, and the seven others are workers. On each core, solving the small-size quadratic subproblems with an efficient software is rather fast (average $0.02s$).
This parallel computing system is thus simple, basic, and homogeneous. In order to reveal the special features of asynchronous algorithms in the experiments, we introduce a small artificial heterogeneity by adding a $0.1s$ waiting time to  {4} scenarios.  {For each experiment, we run each algorithm 10 times and report the median value. In order to display the variability of randomized methods, we also shade the area corresponding to the first and third quartiles for each algorithm.}

\paragraph{Sequential experiments}

{In Figure~\ref{fig:comp1w}}, we compare the Progressive Hedging  (Algorithm\;\ref{alg:PH}, see also remark~\ref{rem:leaf}) with the randomized variant (Algorithm\;\ref{alg:RPH} where we draw 20 scenarios per iteration) with  both uniform sampling and $p$-sampling (see Section~\ref{sec:rand}).  We thus display the decrease of two quantities:
 {\begin{itemize}
    \item the (unconstrained) suboptimality with respect to $\tilde{f} := \sum_{s=1}^S p^s\tilde{f}^s$
    \[
    (\tilde{f}(\tilde{x}^k)-\tilde{f}(x^\star)) / \tilde{f}(x^\star).
    \]
    \item the distance to feasibility, as the distance between $\tilde{\var}^s$ and $\mathcal{C}^s$ over all scenarios
    \[
    \max_{s\in \{1,\ldots,S\}} \|\varvar^{k,s} -\tilde{\var}^{k,s}\|;
    \]
    Note indeed that as for most splitting methods, iterates are asymptotically feasible; more precisely $(\tilde{\var}^k)$ is always in $\NA$ but the individual scenario constraints $\mathcal{C}^s$ are verified  only asymptotically.
\end{itemize}
For illustration purposes, we also provide  the number of subproblems solved along time, and the steplength of the iterates sequence (i.e. the difference between two successive iterates).}

Our first observation is that Progressive Hedging and randomized Progressive Hedging with uniform sampling perform similarly, with respect both to time and to number of subproblems solved (Figures \ref{fig:comp1w_sub} and \ref{fig:comp1w_iter} respectively).
We also notice that $p$-sampling variant gets to a lower suboptimality but with a larger feasibility gap (displayed on Figure\;\ref{fig:comp1w_feas}). This makes sense since scenarios that are prominent in the objective function are also the ones most often drawn and optimized. Conversely, more work needs to be invested on other scenarios to further reduce feasibility.
An interest of the randomized variants is that they produce much more iterates compared to the base algorithm which requires one pass over all scenarios. This can be useful in setups where one iteration of Progressive Hedging is time-consuming.

\begin{figure}[t!]
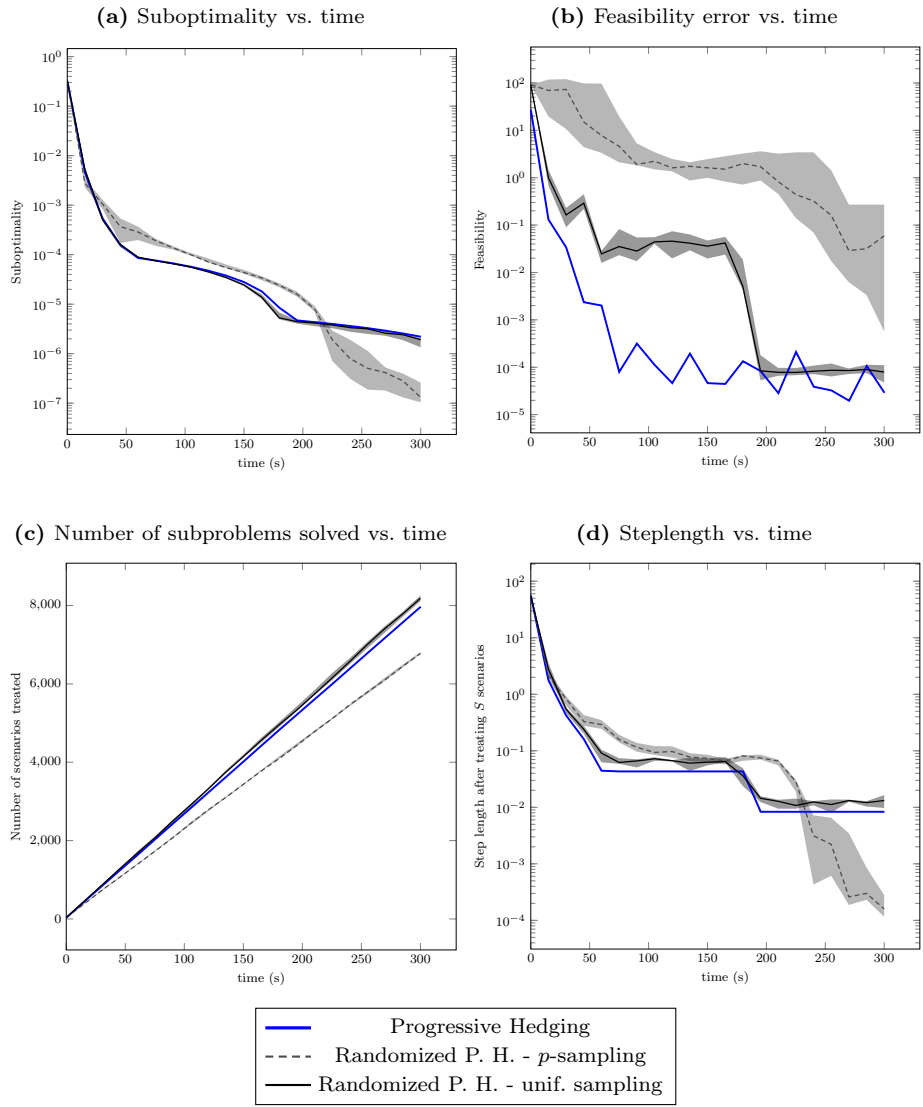

    \centering
    \begin{subfigure}[t]{0.5\textwidth}
        \centering
                \caption{Suboptimality  vs. time \label{fig:comp1w_sub}}
        \resizebox{\textwidth}{!}{
        \input{6stages20dams_seq_subopt.tex}
        }
    \end{subfigure}%
    \begin{subfigure}[t]{0.5\textwidth}
            \centering
                    \caption{Feasibility error  vs. time \label{fig:comp1w_feas}}
                \resizebox{\textwidth}{!}{
        \input{6stages20dams_seq_Cfeas.tex}
                }
    \end{subfigure}\\[0.6cm]
    \begin{subfigure}[t]{0.5\textwidth}
     \centering
       \caption{Number of subproblems solved vs. time \label{fig:comp1w_iter}}
                \resizebox{\textwidth}{!}{
        \input{6stages20dams_seq_nscentreated.tex}
                }
    \end{subfigure}%
    \begin{subfigure}[t]{0.5\textwidth}
  \centering
  \caption{Steplength  vs. time \label{fig:comp1w_steplength}}
        \resizebox{\textwidth}{!}{
        \input{6stages20dams_seq_residual.tex}
        }
    \end{subfigure}\\
       \begin{subfigure}[t]{0.49\textwidth}
        \centering
        \vspace{0.2cm}
        \begin{tikzpicture}
            \begin{customlegend}[
                legend entries={
                    Progressive Hedging,
                    Randomized P. H. - $p$-sampling,
                    Randomized P. H. - unif. sampling,
                }, legend style={font=\footnotesize}]
                \addlegendimage{no marks, solid, color={blue}, very thick}
                \addlegendimage{no marks, thick, densely dashed, color={black!70!white}}
                \addlegendimage{no marks, thick, solid, color={black}}
            \end{customlegend}
        \end{tikzpicture}
        \end{subfigure}
\caption{Comparison of standard vs.\;randomized Progressive Hedging in a sequential set-up.\label{fig:comp1w}}
\end{figure}

\paragraph{Parallel experiments}

After adding seven workers\footnote{In parallel setups, the respective performance of parallel and asynchronous methods is highly variable. We report the experiments obtained on a rather well behaved setup (all workers are equal), still they reflect the general trend we observed.}, we compare in  Figure\;\ref{fig:comp7w} the Parallel Randomized Progressive Hedging (Algorithm\;\ref{alg:RPHPar}) and the Asynchronous Progressive Hedging (Algorithm\;\ref{alg:ARPH}).

We see on Figure\;\ref{fig:comp7w_iter} that with 7 workers,  {Parallel} Randomized Progressive Hedging is able to treat  {about 1.5} times as many scenarios as the sequential methods (randomized or not). Furthermore, the asynchronous variant lifts the communication bottleneck and is able to treat  {4} times as many scenarios as the sequential.

We also see on Figure\;\ref{fig:comp7w_sub} that the Parallel Randomized method converts this higher scenario throughput into efficient iterates: the convergence is faster to the target precision $10^{-8}$ with a similar feasability gap (Figure\;\ref{fig:comp7w_feas}). Thus, this variant is a simple and efficient way to solve multistage problems on parallel setups.

A final remark from Figure\;\ref{fig:comp7w_sub} is that the theoretical stepsize of Theorem\;\ref{th:arph} (obtained by taking the maximum observed delay) is overly pessimistic, resulting in a slow algorithm. Taking a unit stepsize performs well for this instance. However, we observed in other setups that a unit stepsize may lead to non-convergence; in general, some tuning of this parameter is required for better performance.

\begin{figure}[t!]
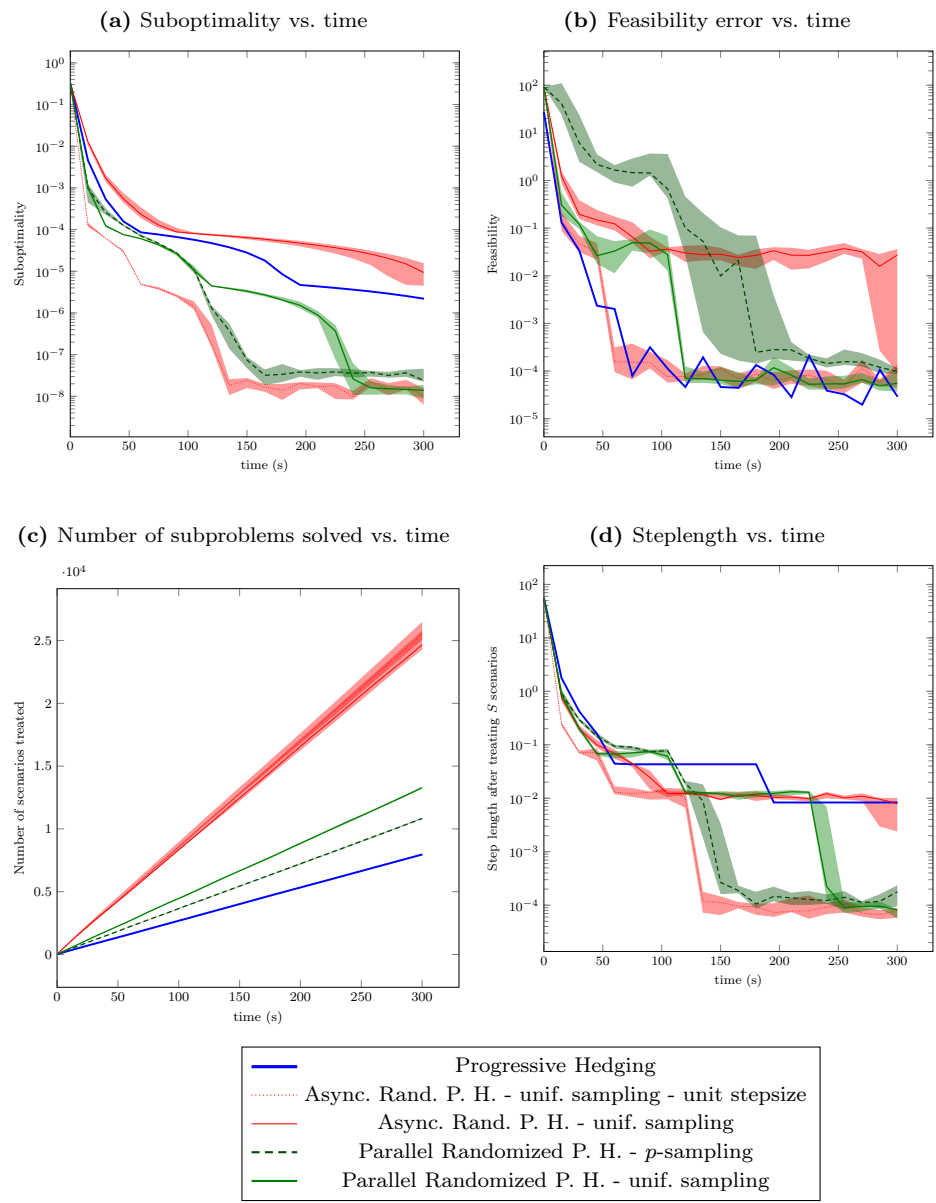

    \centering
    \begin{subfigure}[t]{0.5\textwidth}
        \centering
                    \caption{Suboptimality  vs. time \label{fig:comp7w_sub}}
        \resizebox{\textwidth}{!}{
        \input{6stages20dams_dist_subopt.tex}
        }
    \end{subfigure}%
    ~
    \begin{subfigure}[t]{0.5\textwidth}
            \centering
                  \caption{Feasibility error  vs. time \label{fig:comp7w_feas}}
                \resizebox{\textwidth}{!}{
        \input{6stages20dams_dist_Cfeas.tex}
                }
    \end{subfigure}\\[0.6cm]
    \begin{subfigure}[t]{0.5\textwidth}
         \centering
                 \caption{Number of subproblems solved vs. time \label{fig:comp7w_iter}}
                \resizebox{\textwidth}{!}{
        \input{6stages20dams_dist_nscentreated.tex}
                }
    \end{subfigure}%
    ~
    \begin{subfigure}[t]{0.5\textwidth}
\centering
 \caption{Steplength  vs. time \label{fig:comp7w_steplength}}
        \resizebox{\textwidth}{!}{
        \input{6stages20dams_dist_residual.tex}
        }
    \end{subfigure} \\
    \begin{subfigure}[t]{0.49\textwidth}
        \centering
        \vspace{0.2cm}
        \begin{tikzpicture}
            \begin{customlegend}[
                legend entries={
                    Progressive Hedging,
                    Async. Rand. P. H. - unif. sampling - unit stepsize,
                    Async. Rand. P. H. - unif. sampling,
                    Parallel Randomized P. H. - $p$-sampling,
                    Parallel Randomized P. H. - unif. sampling,
                }, legend style={font=\footnotesize}]
                \addlegendimage{no marks, solid, color={blue}, very thick}
                \addlegendimage{no marks, color={red}, densely dotted}
                \addlegendimage{no marks, solid, color={red}}
                \addlegendimage{no marks, thick, densely dashed, color={green!30!black}}
                \addlegendimage{no marks, thick, solid, color={green!50!black}}
            \end{customlegend}
        \end{tikzpicture}
    \end{subfigure}
\caption{Comparison of standard vs.\;randomized Progressive Hedging in a parallel set-up with $7$ workers.\label{fig:comp7w}}
\end{figure}

\section*{Acknowledgments}

The authors wish to thank the associate editor and the two anonymous reviewers for their valuable comments, notably with respect to the placement in the literature, which greatly improved the paper. 
F.I and J.M. thank Welington de Oliveira for fruitful discussions at the very beginning of this project.

\ifdraft
    \bibliographystyle{ormsv080}
\else
    \bibliographystyle{spmpsci}
\fi
\bibliography{optim}

\appendix

\section{Fixed-point view  {of} Progressive Hedging}\label{sec:OPPH}

This appendix complements Section\;\ref{sec:standard}: we reformulate the multistage problem \eqref{eq:pb} as finding a fixed point of some operator (see the textbook\;\cite[Chap.~3.9]{ruszczynski2003stochastic}). For all definitions and results on monotone operator theory, we refer to \cite{bauschke2011convex}.

Denoting the objective function by $f(\var):=\sum_{s=1}^S p_s f^s(\var)$ and the indicator of constraints by $\iota_{\NA}$ with $\iota_{\NA}(\var) = 0$ if $\var\in\NA$ and $+\infty$ otherwise, we have that solving \eqref{eq:pb} amounts to finding $\var^\star$ such that
\[
    0 \in \partial (f + \iota_{\NA})(\var^\star) = \partial f(\var^\star) + \partial \iota_{\NA}(\var^\star)
\]
where we use\;Assumption\;\ref{hyp:gen2} for the equality. Then, we introduce the two following operators
\begin{align}\label{eq:op}
    \A(\var) :=  P^{-1} \partial f(\var) ~~\text{ and }~~  \B(\var) := P^{-1} \partial  \iota_{\NA}(\var)
\end{align}
where $P = \mathrm{diag}( p_1,\ldots,p_S )$. Using Assumption~\ref{hyp:gen}, the operators $\A$ and $\B$ defined in \eqref{eq:op} are maximal monotone since so are the subdifferentials of convex proper lower-semicontinuous functions.

Solving \eqref{eq:pb} thus amounts to finding a zero of $\A+\B$ the sum of two maximal monotone operators:
\begin{equation}\label{eq:zero}
\text{$\var^\star$ solves \eqref{eq:pb}} \iff  \text{$\var^\star$ is a zero of $\A + \B$ i.e.}~ 0\in\A(\var^\star) + \B(\var^\star).
\end{equation}

We follow the notation of \cite[Chap.\;3]{ruszczynski2003stochastic} and the properties of \cite[Chap.~4.1 and 23.1]{bauschke2011convex}. For a given maximal monotone operator $\M$, we define for any $\mu>0$ two associated operators:
\begin{itemize}
        \item[i)] the resolvent $\mathsf{J}_{\mu\M} = (I+\mu \M)^{-1} $, (which is well-defined and firmly non-expansive),
        \item[ii)] the reflected resolvent $\mathsf{O}_{\mu\M} = 2\mathsf{J}_{\mu\M}-I$ (which is non-expansive).
\end{itemize}

These operators allow us to formulate our multistage problem as a fixed-point problem: with the help of \eqref{eq:zero} and \cite[Prop.\;25.1(ii)]{bauschke2011convex}, we have
\[
\text{$\var^\star$ solves \eqref{eq:pb}} \iff  \text{$\var^\star = \mathsf{J}_{\mu\B}(\varglob^\star)$  with $\varglob^\star$ a fixed point of $\mathsf{O}_{\mu\A}\circ\mathsf{O}_{\mu\B}$, i.e. $\varglob^\star=\mathsf{O}_{\mu\A}\circ\mathsf{O}_{\mu\B}(\varglob^\star)$.}
\]

We can apply now a fixed-point algorithm to the firmly non-expansive operator\footnote{As $\mathsf{O}_{\mu\A}$  and $\mathsf{O}_{\mu\B}$ are non-expansive but not firmly non-expansive, it is necessary to average them with the current iterate (this is often called the Krasnosel'ski\u{\i}--Mann algorithm \cite[Chap 5.2]{bauschke2011convex}) to make this iteration firmly non-expansive and ensure Fejér monotone convergence.} $\frac{1}{2}\mathsf{O}_{\mu\A} \circ \mathsf{O}_{\mu\B} + \frac{1}{2}\Id $ to find a fixed point of $\mathsf{O}_{\mu\A} \circ \mathsf{O}_{\mu\B}$.

This gives the following iteration (equivalent to Douglas-Rachford splitting)
\begin{align}
    \label{eq:DR}
    \varglob^{k+1} = \frac{1}{2}\mathsf{O}_{\mu\A}(\mathsf{O}_{\mu\B}(\varglob^k)) + \frac{1}{2}\varglob^k
\end{align}
which converges to a point $\varglob^\star$ such that $x^\star := \mathsf{J}_{\mu\B}(\varglob^\star) $ is a zero of $\A+\B$;
see \cite[Chap.~25.2]{bauschke2011convex}.

It is well-known (see e.g.\;the textbook\;\cite[Chap.~3,~Fig.~10]{ruszczynski2003stochastic}) that this algorithm with the operators $\A$ and $\B$ defined in \eqref{eq:op} leads to the Progressive Hedging algorithm. We give here a short proof of this property; along the way, we introduce basic properties and arguments used in the new developments on randomized Progressive Hedging of the next two appendices. We provide first the expressions of the reflected resolvent operators for $\A$ and $\B$.

\begin{lemma}[Operators associated with Progressive Hedging]\label{lem:AB}
Let endow the space $\E$ of $S\times n$ real matrices with the weighted inner product  $\langle A,B\rangle_P = \mathrm{Trace}(A^\mathrm{T} P B)$. Then the operators $\A$ and $\B$ defined in \eqref{eq:op} are maximal monotone, and their reflected resolvent operators $\mathsf{O}_{\mu\A}$ and $\mathsf{O}_{\mu\B}$  have the following expressions:
\begin{itemize}
    \item[i)] $\mathsf{O}_{\mu\A}(\varglob)  = \var - \mu \vartwo $ with
      \begin{align*}
         \var^s = \argmin_{y\in\RR^n}\left\{ f^s(y) + \frac{1}{2\mu} \left\|y-\varglob^s \right\|^2  \right\} \text{ for all } s=1,\ldots,S
      \end{align*}
      and $\vartwo =(\varglob-\var)/\mu $ (hence $\mathsf{O}_{\mu\A}(\varglob) = 2\var - \varglob$);
      \smallskip
    \item[ii)] $\mathsf{O}_{\mu\B}(\varglob)  = \var - \mu \vartwo $ with
    \begin{align*}
       \var_t^s = \frac{1}{\sum_{\sigma\in\mathcal{B}^s_t} p_\sigma } \sum_{\sigma\in\mathcal{B}^s_t} p_\sigma \varglob_t^\sigma  \text{ for all } s=1,\ldots,S \text{ and } t=1,\ldots,T
    \end{align*}
    and $\vartwo =(\varglob-\var)/\mu $ (hence $\mathsf{O}_{\mu\B}(\varglob) = 2\var - \varglob$). The point $\var$ is the \emph{orthogonal projection} of $\varglob$ to $\NA$. Thus, $\varglob$ writes uniquely as $\varglob = \var + \mu \vartwo $ with $\var\in\NA$ and $\vartwo\in\NA^\perp$.
\end{itemize}
\end{lemma}

\begin{proof}
Since $\partial f(\cdot)$ and $\partial  \iota_{\NA}(\cdot)$ are the subdifferentials of convex proper lower-semicontinuous functions, they are maximal monotone with respect to the usual inner product, and there so are $\A$ and $\B$, with respect to the weighted inner product.

Applying \cite[Prop.\;23.1]{bauschke2011convex} to a maximal monotone operator $\M$, we get that $\varglob\in\E$ can be uniquely represented as $\varglob = \var + \mu \vartwo$ with  $\vartwo\in\M(\var)$, thus $\mathsf{J}_{\mu\M}(\varglob) = \var$ and $\mathsf{O}_{\mu\M}(\varglob) = \mathsf{O}_{\mu\M}(\var + \mu \vartwo) = \var - \mu \vartwo$. This gives the expressions for $\mathsf{O}_{\mu\A}$ and $\mathsf{O}_{\mu\B}$ from the expressions of $\mathsf{J}_{\mu\A}$ and  $\mathsf{J}_{\mu\B}$ based on the proximity operators associated with $f$ and $\iota_{\NA}$ (see \cite[Prop.\;16.34]{bauschke2011convex}).
\end{proof}

We now apply the general Douglas-Rachford scheme \eqref{eq:DR} with the expressions obtained in Lemma~\ref{lem:AB}. We first get:
\begin{align*}
    \left\{
        \begin{array}{ll}
            \var_t^{k,s} = \frac{1}{\sum_{\sigma\in\mathcal{B}^s_t} p_\sigma } \sum_{\sigma\in\mathcal{B}^s_t} p_\sigma \varglob_t^{k,\sigma}  \text{ for all } s=1,\ldots,S \text{ and } t=1,\ldots,T  & {\scriptstyle \var^k\in\NA}\\[2ex]
            \varvarglob^k = \mathsf{O}_{\mu\B}(\varglob^k) = 2\var^k - \varglob^k = \var^k - \mu \vartwo^k  & {\scriptstyle  \text{ with } \vartwo^k = (\varglob^k-\var^k)/\mu \in \NA^\perp  } \\
            & {\scriptstyle  \text{ thus } \vartwo^k = \vartwo^{k-1} + \frac{1}{\mu} (\varvar^k-\var^k)  } \\
            \varvar^{k+1,s} = \argmin_{y\in\RR^n}\left\{ f^s(y) + \frac{1}{2\mu} \left\|y-\varvarglob^{k,s} \right\|^2  \right\} \text{ for all } s=1,\ldots,S & \\
            \varglob^{k+1} = \frac{1}{2} (2\varvar^{k+1}-\varvarglob^k) + \frac{1}{2}\varglob^k = \varglob^k + \varvar^{k+1} - \var^{k+1} = \varvar^{k+1} + \mu \vartwo^k &
        \end{array}
    \right.
\end{align*}
Let us reorganize the equations and eliminate intermediate variables. In particular, we use the fact that, provided that the algorithm is initialized with $\var^0 \in \NA$ and $\vartwo^0\in\NA^\perp$, all iterates $(\var^k)$ and $(\vartwo^k)$ are in $\NA$ and $\NA^\perp$ respectively. We eventually obtain:
\begin{align*}
    \left\{
        \begin{array}{ll}
            \varvar^{k+1,s} = \argmin_{y\in\RR^n}\left\{ f^s(y) + \frac{1}{2\mu} \left\|y-\var^{k,s} + \mu \vartwo^{k,s} \right\|^2  \right\} \text{ for all } s=1,\ldots,S & \\[2ex]
            \var_t^{k+1,s} = \frac{1}{\sum_{\sigma\in\mathcal{B}^s_t} p_\sigma } \sum_{\sigma\in\mathcal{B}^s_t} p_\sigma \varvar_t^{k+1,\sigma}  \text{ for all } s=1,\ldots,S \text{ and } t=1,\ldots,T & \\ \hfill {\scriptstyle \var^k\in\NA \text{ converges to a solution of \eqref{eq:pb}}} & \\
            \vartwo^{k+1} = \vartwo^{k} + \frac{1}{\mu} (\varvar^{k+1}-\var^{k+1}) &
        \end{array}
    \right.
\end{align*}
This is exactly the Progressive Hedging algorithm, written with similar notation as in the textbook\;\cite[Chap.~3,~Fig.~10]{ruszczynski2003stochastic}. The convergence of the algorithm (recalled in Theorem\;\ref{th:ph})) can be obtained directly by instantiating the general convergence result of the Douglas-Rachford method
\cite[Chap.~25.2]{bauschke2011convex}.

In the next two appendices, we are going to follow the same line that has brought us from Douglas-Rachford to Progressive Hedging, to go from \emph{randomized} Douglas-Rachford to \emph{randomized} Progressive Hedging, and from \emph{asynchronous} Douglas-Rachford to \emph{asynchronous} Progressive Hedging.

\section{Derivation and Proof of the Randomized Progressive Hedging}\label{apx:rph}

A randomized counterpart of the Douglas-Rachford method\;\eqref{eq:DR} consists in updating only part of the variable chosen at random; see \cite{iutzeler2013asynchronous} and extensions \cite{bianchi2015coordinate,combettes2015stochastic}. At each iteration, this variant amounts to update the variables corresponding to the chosen scenario $\selected$ (randomly chosen with probability $q_\selected$), the other staying unchanged:
\begin{align}
\nonumber
    &\text{Draw a scenario } \selected \in \{1,\ldots,S\} \text{ with probability } \mathbb{P}[\selected = s] = q_s  \\
    & \left| \begin{array}{l}
        \varglob^{k+1,\selected} = \frac{1}{2}\left[\mathsf{O}_{\mu\A}(\mathsf{O}_{\mu\B}(\varglob^k))\right]^{\selected} + \frac{1}{2}\varglob^{k,\selected} \\
    \varglob^{k+1,s} =  \varglob^{k,s} \text{ for all } s\neq \selected
    \end{array}
     \right.
         \label{eq:RDR}
\end{align}

 {
Our goal is to obtain the Randomized Progressive Hedging (Algorithm\;\ref{alg:RPH}) as an instantiation of \eqref{eq:RDR} with the operators defined in Lemma\;\ref{lem:AB} in Appendix~\ref{sec:OPPH}. Before proceeding with the derivation, let us prove the convergence of \eqref{eq:RDR} with these operators. %

\begin{proposition}
\label{th:RDR}
Consider a multistage problem \eqref{eq:pb} verifying Assumptions\;\ref{hyp:gen}\;and\;\ref{hyp:gen2}.  Then, the sequence $(\varglob^k)$ generated by \eqref{eq:RDR} with  $\mathsf{O}_{\mu\A}$ and $\mathsf{O}_{\mu\B}$ defined in Lemma~\ref{lem:AB} converges almost surely to a fixed point of  $\mathsf{O}_{\mu\A}\circ\mathsf{O}_{\mu\B}$. Furthermore, $\tilde{\var}^k := \mathsf{J}_{\mu\B}(\varglob^k)$ converges to a solution of~\eqref{eq:pb}.
\end{proposition}

\begin{proof}
First, recall from Lemma\;\ref{lem:AB} that under assumptions\;\ref{hyp:gen}\;and\;\ref{hyp:gen2}, the operators $\A,\B$ of \eqref{eq:op} are maximal monotone. Then, the associated operators $\mathsf{O}_{\mu\A}$ and $\mathsf{O}_{\mu\A}$ are then non-expansive by construction  (see \cite[Chap.\;4.1]{bauschke2011convex}), and therefore the iteration $\mathsf{T} = (\mathsf{O}_{\mu\A}\circ \mathsf{O}_{\mu\B} + I)/2 $ is firmly non expansive. This is the key assumption to use the convergence result \cite[Th.~2]{iutzeler2013asynchronous} which gives that the sequence $(\varglob^k)$ generated by \eqref{eq:RDR} converges almost surely to a fixed point of $\mathsf{O}_{\mu\A}\circ\mathsf{O}_{\mu\B}$. Using the continuity of $ \mathsf{J}_{\mu\B}$ and the fact that $x^\star := \mathsf{J}_{\mu\B}(\varglob^\star) $ is a zero of  $\A+\B$ (i.e. solves the multi-stage problem  \eqref{eq:pb} by \eqref{eq:zero}) gives the last part of the result.
\end{proof}

Now that the convergence of \eqref{eq:RDR} with the operators of Appendix~\ref{sec:OPPH} has been proven, let us derive our Randomized Progressive Hedging (Algorithm\;\ref{alg:RPH}) as an equivalent formulation of \eqref{eq:RDR}. By doing so, the associated convergence result (Theorem\;\ref{th:rph}) directly follows from Proposition~\ref{th:RDR}.
}

From the specific expressions of operators $\mathsf{O}_{\mu\A}$ and $\mathsf{O}_{\mu\B}$ (Lemma~\ref{lem:AB}), we see that these operators are very different in nature:
\begin{itemize}
    \item $\mathsf{O}_{\mu\A}$ \emph{is separable by scenario} but involves \emph{solving a subproblem};
    \item $\mathsf{O}_{\mu\B}$ \emph{links the scenarios} but only amounts to \emph{computing a weighted average}.
\end{itemize}
To leverage this structure, we apply the randomized Douglas-Rachford method \eqref{eq:RDR} and get:
\begin{align*}
    \left\{
        \begin{array}{ll}
           \text{Draw a scenario } \selected \in \{1,\ldots,S\} \text{ with probability } \mathbb{P}[\selected = s] = q_s  & \\
            \var_t^{k,s} = \frac{1}{\sum_{\sigma\in\mathcal{B}^s_t} p_\sigma } \sum_{\sigma\in\mathcal{B}^s_t} p_\sigma \varglob_t^{k,\sigma}  \text{ for all } s=1,\ldots,S \text{ and } t=1,\ldots,T \hfill {\scriptstyle \var^k\in\NA} & \\
            \varvarglob^k = \mathsf{O}_{\mu\B}(\varglob^k) = 2\var^k - \varglob^k = \var^k - \mu \vartwo^k \hfill {\scriptstyle  \text{ with } \vartwo^k = (\varglob^k-\var^k)/\mu \in \NA^\perp  } &  \\
            \varvar^{k+1,s} = \argmin_{y\in\RR^n}\left\{ f^s(y) + \frac{1}{2\mu} \left\|y-\varvarglob^{k,s} \right\|^2  \right\} \text{ for all } s=1,\ldots,S & \\
            \left| \begin{array}{l}
                \varglob^{k+1,\selected} =\frac{1}{2} (2\varvar^{k+1,\selected}-\varvarglob^{k,\selected}) + \frac{1}{2}\varglob^{k,\selected} = \varglob^{k,\selected} + \varvar^{k+1,\selected} - \var^{k+1,\selected} = \varvar^{k+1,\selected} + \mu \vartwo^{k,\selected} \\
            \varglob^{k+1,s} =  \varglob^{k,s} \text{ for all } s\neq \selected
            \end{array}
             \right.
            &
        \end{array}
    \right.
\end{align*}
Let us carefully prune unnecessary computations. First, only $\varvar^{k+1,\selected}$ needs to be computed, so the other $\varvar^{k+1,s}$ ($s\neq \selected$) can be safely dropped. The same holds for $\var^{k,\selected}$, $\varvarglob^{k,\selected}$, and $\vartwo^{k,\selected}$. However, even though only $\var^{k,\selected}$ need to be computed, it depends on all the other scenarios through the projection operator, so the iterates have to be computed successively and with only a partial update of $\vartwo^k$ (in contrast with Appendix\;\ref{sec:OPPH}, $\vartwo^k$ does not belong to $\NA$ anymore and thus cannot be dropped out of the projection, thus we keep directly the global variable $\varglob^k$ updated):
\begin{align*}
    \left\{
        \begin{array}{ll}
           \text{Draw a scenario } \selected \in \{1,\ldots,S\} \text{ with probability } \mathbb{P}[\selected = s] = q_s  & \\
            \var_t^{k,\selected} = \frac{1}{\sum_{\sigma\in\mathcal{B}^{\selected}_t} p_\sigma } \sum_{\sigma\in\mathcal{B}^{\selected}_t} p_\sigma \varglob_t^{k,\sigma}  \text{ for all }  t=1,\ldots,T  & \\
            \varvarglob^{k,\selected} =  2\var^{k,\selected} - \varglob^{k,\selected}  &   \\
            \varvar^{k+1,\selected} = \argmin_{y\in\RR^n}\left\{ f^{\selected}(y) + \frac{1}{2\mu} \left\|y-\varvarglob^{k,\selected} \right\|^2  \right\}&  \\
            \left| \begin{array}{l}
                \varglob^{k+1,\selected} = \varglob^{k,\selected} + \varvar^{k+1,\selected} - \var^{k+1,\selected}  \\
            \varglob^{k+1,s} =  \varglob^{k,s} \text{ for all } s\neq \selected
            \end{array}
             \right.
            &
        \end{array}
    \right.
\end{align*}
Eliminating intermediate variable $ \varvarglob$, we obtain the randomized Progressive Hedging:
\begin{align*}
    \left\{
        \begin{array}{ll}
           \text{Draw a scenario } \selected \in \{1,\ldots,S\} \text{ with probability } \mathbb{P}[\selected = s] = q_s  & \\
            \var_t^{k+1,\selected} = \frac{1}{\sum_{\sigma\in\mathcal{B}^{\selected}_t} p_\sigma } \sum_{\sigma\in\mathcal{B}^{\selected}_t} p_\sigma \varglob_t^{k,\sigma}  \text{ for all }  t=1,\ldots,T  & \\
            \varvar^{k+1,\selected} = \argmin_{y\in\RR^n}\left\{ f^{\selected}(y) + \frac{1}{2\mu} \left\|y- 2\var^{k+1,\selected} + \varglob^{k,\selected} \right\|^2  \right\}& \\
            \left| \begin{array}{l}
                \varglob^{k+1,\selected} = \varglob^{k,\selected} + \varvar^{k+1,\selected} - \var^{k+1,\selected}  \\
            \varglob^{k+1,s} =  \varglob^{k,s} \text{ for all } s\neq \selected
            \end{array}
             \right.
            &
        \end{array}
    \right.
\end{align*}
 {
Finally, notice that from Proposition~\ref{th:RDR}, that the variable converging to a solution of \eqref{eq:pb} is  $\tilde{\var}^k := \mathsf{J}_{\mu\B}(\varglob^k)$. From Lemma~\ref{lem:AB} (and the fact that $\mathsf{O}_{\mu\B} = 2\mathsf{J}_{\mu\B} - \mathsf{I}$), we get that  $\tilde{\var}_t^{k,s} = \frac{1}{\sum_{\sigma\in\mathcal{B}^s_t} p_\sigma } \sum_{\sigma\in\mathcal{B}^s_t} p_\sigma \varglob_t^{k,\sigma}  \text{ for all } s=1,\ldots,S \text{ and } t=1,\ldots,T$
 and that $\tilde{\var}^k \in \NA$. }

\section{Derivation and Proof of the Asynchronous Randomized Progressive Hedging}\label{apx:arph}

Using again the bridge between Progressive Hedging and fixed-point algorithms, we present here how to derive an asynchronuous progressive hedeging from the asynchronous parallel fixed-point algorithm ARock\;\cite{peng2016arock}. In order to match the notation and derivations of \cite{peng2016arock}, let us define the operator  $\mathsf{S} := I - \mathsf{O}_{\mu\A} \circ \mathsf{O}_{\mu\B}$, the zeros of which coincide with the fixed points of $\mathsf{O}_{\mu\A} \circ \mathsf{O}_{\mu\B}$. Applying ARock to this operator leads to the following iteration:
\begin{align}
\nonumber
& \text{Every worker asynchronously do} \\
& \left\{
\begin{array}{l}
    \text{Draw a scenario } \selected  \in \{1,\ldots,S\} \text{ with probability } \mathbb{P}[\selected = s] = q_s \\
    \left| \begin{array}{l}
        \varglob^{k+1,\selected} = \varglob^{k,\selected} - \frac{\eta^k}{S p_{s^k}} \left( \varglobold^{k,s^k}  - \left[\mathsf{O}_{\mu\A}(\mathsf{O}_{\mu\B}(\varglobold^k))\right]^{\selected} \right) \\
    \varglob^{k+1,s} =  \varglob^{k,s} \text{ for all } s\neq \selected
    \end{array}
     \right. \\
     {\scriptstyle
     \text{where } \varglobold^k \text{ is the value of }  \varglob^{k} \text{ used by the updating worker at time $k$ for its computation}}
\end{array}
\right. \label{eq:adr}
\end{align}
Notice that the main difference between this iteration and \eqref{eq:RDR} is the introduction of the variable $\varglobold^k$ which is used to handle delays between workers in asynchronous computations:
\begin{itemize}
    \item If there is only one worker, it just computes its new point with the latest value so we simply have: $\varglobold^k = \varglob^k$. We notice that taking $\eta^k = S p_{s^k} /2 $, we recover exactly the randomized Douglas-Rachford method \eqref{eq:RDR};
    \smallskip
    \item If there are several workers, $\varglobold^k $ is usually an older version of the main variable, as other workers may have updated the main variable during the computation of the updating worker. In this case, we have $\varglobold^k  =  \varglob^{k-d^k}$ where $d^k$ is the delay suffered by the updating worker at time~$k$.
\end{itemize}

 {
We derive here our Asynchronous Randomized Progressive hedging (Algorithm\;\ref{alg:ARPH}) as an instantiation of \eqref{eq:adr} with the operators $\mathsf{O}_{\mu\A}$ and $\mathsf{O}_{\mu\B}$ defined in Appendix~\ref{sec:OPPH}. Let us establish first the convergence of this scheme using a general result of\;\cite{peng2016arock} which makes little assumptions on the communications between workers and master. The main requirement is that the maximum delay between workers is bounded, which is a reasonable assumption when the algorithm is run on a multi-core machine or on a medium-size computing cluster.

\begin{proposition}
\label{th:ADR}
Consider a multistage problem \eqref{eq:pb} verifying Assumptions\;\ref{hyp:gen}\;and\;\ref{hyp:gen2}.
We assume furthermore that the delays are bounded:
$d^k\leq \tau <\infty $ for all $k$.
If we take the stepsize $\eta^k$ as follows for some fixed  $0<c<1$
\begin{equation}\label{eq:bounddelay}
0< \eta_{\min} \leq \eta^k ~\leq~  \frac{c S q_{\min}}{2\tau \sqrt{q_{\min}} +1 }
\qquad \text{with $q_{\min} = \min_s q_s$}.
\end{equation}
Then, the sequence $(\varglob^k)$ generated by \eqref{eq:adr} with  $\mathsf{O}_{\mu\A}$ and $\mathsf{O}_{\mu\B}$ defined in Lemma~\ref{lem:AB} converges almost surely to a fixed point of  $\mathsf{O}_{\mu\A}\circ\mathsf{O}_{\mu\B}$. Furthermore, $\tilde{\var}^k := \mathsf{J}_{\mu\B}(\varglob^k)$ converges to a solution of~\eqref{eq:pb}.
\end{proposition}

\begin{proof}
The beginning of the proof follows the same lines as the one of Proposition~\ref{th:RDR} to show that $\mathsf{O}_{\mu\A}$ and $\mathsf{O}_{\mu\A}$ are non-expansive by construction, which implies that  $\mathsf{S} := I - \mathsf{O}_{\mu\A} \circ \mathsf{O}_{\mu\B}$ is also non-expansive with its zeros corresponding to the fixed points of $\mathsf{O}_{\mu\A} \circ \mathsf{O}_{\mu\B}$ (see \cite[Chap.~4.1]{bauschke2011convex}).
We can then apply \cite[Th.~3.7]{peng2016arock} to get that $(\varglob^k)$ converges almost surely to a zero of $\mathsf{S}$. As in the proof of Proposition~\ref{th:RDR}, we use the continuity of $ \mathsf{J}_{\mu\B}$ and the fact that $x^\star := \mathsf{J}_{\mu\B}(\varglob^\star) $ is a zero of $\A+\B$ (i.e. solves the multi-stage problem  \eqref{eq:pb} by \eqref{eq:zero}) to get the last part of the result.
\end{proof}
}

Using the expressions of the operators of Lemma\;\ref{lem:AB}, \eqref{eq:adr} writes
\begin{align*}
& \text{Every worker asynchronously do} \\
&    \left\{
        \begin{array}{ll}
           \text{Draw a scenario } \selected \in \{1,\ldots,S\} \text{ with probability } \mathbb{P}[\selected = s] = q_s  & \\
            \varold_t^{k,s} = \frac{1}{\sum_{\sigma\in\mathcal{B}^s_t} p_\sigma } \sum_{\sigma\in\mathcal{B}^s_t} p_\sigma \varglobold_t^{k,\sigma}  \text{ for all } s=1,\ldots,S \text{ and } t=1,\ldots,T  & \\
            \varvarglobold^k = \mathsf{O}_{\mu\B}(\varglobold^k) = 2\varold^k - \varglobold^k  & \\
            \varvarold^{k+1,s} = \argmin_{y\in\RR^n}\left\{ f^s(y) + \frac{1}{2\mu} \left\|y-\varvarglobold^{k,s} \right\|^2  \right\} \text{ for all } s=1,\ldots,S & \\
            \left[\mathsf{O}_{\mu\A}(\mathsf{O}_{\mu\B}(\varglobold^k))\right]^{\selected} = 2 \varvarold^{k+1,\selected} - \varvarglobold^{k,\selected}  & \\
            \left| \begin{array}{l}
                \varglob^{k+1,\selected} = \varglob^{k,\selected} - \frac{\eta^k}{S p_{s^k}} \left( \varglobold^{k,s^k}  - \left[\mathsf{O}_{\mu\A}(\mathsf{O}_{\mu\B}(\varglobold^k))\right]^{\selected} \right)  \\
            \varglob^{k+1,s} =  \varglob^{k,s} \text{ for all } s\neq \selected
            \end{array}
             \right.
            &
        \end{array}
    \right.
\end{align*}
Pruning unnecessary computations, the asynchronous version of Progressing Hedging boils down to:
\begin{align*}
& \text{Every worker asynchronously do} \\
&    \left\{
        \begin{array}{ll}
           \text{Draw a scenario } \selected \in \{1,\ldots,S\} \text{ with probability } \mathbb{P}[\selected = s] = q_s  & \\
            \varold_t^{k,\selected} = \frac{1}{\sum_{\sigma\in\mathcal{B}^{\selected}} p_\sigma } \sum_{\sigma\in\mathcal{B}^{\selected}_t} p_\sigma \varglobold_t^{k,\sigma}  \text{ for all }  t=1,\ldots,T  & \\
            \varvarold^{k+1,\selected} = \argmin_{y\in\RR^n}\left\{ f^{\selected}(y) + \frac{1}{2\mu} \left\|y- 2\varold^{k,\selected} + \varglobold^{k,\selected} \right\|^2  \right\} &  \\
            \left| \begin{array}{l}
                \varglob^{k+1,\selected} = \varglob^{k,\selected} + \frac{2 \eta^k}{S p_{s^k}} \left(  \varvarold^{k+1,\selected} -  \varold^{k,\selected} \right)  \\
            \varglob^{k+1,s} =  \varglob^{k,s} \text{ for all } s\neq \selected
            \end{array}
             \right.
            &
        \end{array}
    \right.
\end{align*}
This asynchronous algorithm can be readily rewritten as Algorithm\;\ref{alg:ARPH}, highlighting the master-worker implementation.  {Theorem\;\ref{th:rph} then follows directly from Proposition\;\ref{th:ADR}}.

\section{\tool toolbox: Implementations Details}\label{sec:rph_details}

A basic presentation of the toolbox \tool\;is provided in Section\;\ref{sec:toolbox}; a complete description is available on the online documentation. In this section, we briefly provide complementary information on the input/output formats.

The input format is a \texttt{Julia} structure, named \texttt{problem}, that gathers all the information to solve a given multi-stage problem.

\begin{lstlisting}
struct Problem{T} # Main input class implemented in src files
    scenarios::Vector{T}
    build_subpb::Function
    probas::Vector{Float64}
    nscenarios::Int
    nstages::Int
    stage_to_dim::Vector{UnitRange{Int}}
    scenariotree::ScenarioTree
end
\end{lstlisting}
The attribute \texttt{scenarios} is an array representing the possible scenarios of the problem. \texttt{nscenarios} is the total number of scenarios brought by the user and the probability affected to each scenario is indicated by the attribute \texttt{probas}. The number of stages, assumed to be equal among all scenarios, is stored in the attribute \texttt{nstages}. The dimension of the variable associated to each stage is stored in the vector of couples \texttt{stage\_to\_dim}: if for a fixed stage $i$, $\texttt{stage\_to\_dim[i]} = \texttt{p:q}$, then the variable associated to stage $i$ is of dimension $q - p + 1$. Each of the scenarios must inherit the abstract structure \texttt{AbstractScenario}. This abstract structure does not impose any requirements on the scenarios themselves, so that the user is free to plug any relevant information in these scenarios. Here is an example.

\begin{lstlisting}
abstract type AbstractScenario end  ## Abstract class implemented in src files

struct UserScenario <: AbstractScenario ## Custom class to be designed by the user
    trajcenter::Vector{Float64}
    constraintbound::Int
end

## Class Atributes to be designed by the user
# Instantiation of 4 scenarios
scenario1 = UserScenario([1, 1, 1], 3)
scenario2 = UserScenario([2, 2, 2], 3)
scenario3 = UserScenario([3, 3, 3], 3)
scenario4 = UserScenario([3, 3, 3], 3)
custom_scenarios = [scenario1, scenario2, scenario3, scenario4]

custom_nscenarios = length(custom_scenarios)
custom_stage_to_dim = [1:1, 2:2]
custom_nstages = length(stage_to_dim)
probabilities = [0.1, 0.25, 0.50, 0.15]
\end{lstlisting}

The function \texttt{build\_subpb}, provided by the user, informs the solver about the objective function $f_s$ to use for each scenario. This function is assumed to take as inputs a \texttt{Jump.model} object, a single \texttt{scenario} object and an object \texttt{scenarioId}, which corresponds to an integer that identifies the  scenario. \texttt{build\_subpb} must then return the variable designed for the optimization, named \texttt{y} below, an expression of the objective function $f_s$, denoted below \texttt{objexpr} as well as the constraints relative to this scenario, denoted below \texttt{ctrref}.

\begin{lstlisting}
# Function to be designed by the user
function custom_build_subpb(model::JuMP.Model, s::UserScenario, id_scen::ScenarioId)
    n = length(s.trajcenter)

    y = @variable(model, [1:n], base_name="y_s"*string(id_scen))
    objexpr = sum((y[i] - s.trajcenter[i])^2 for i in 1:n)
    ctrref = @constraint(model, y .<= s.constraintbound)

    return y, objexpr, ctrref
end
\end{lstlisting}

Finally, the attribute \texttt{scenariotree} is aimed at storing the graph structure of the scenarios. \texttt{scenariotree} must be of type \texttt{ScenarioTree}, a tree structure designed by the authors. One can build an object scenario tree, by directly stating the shape of the tree with the help of \texttt{Julia} set structure.
\begin{lstlisting}
## Instantiation of the scenario tree
stageid_to_scenpart = [
        OrderedSet([BitSet(1:4)]),                      # Stage 1
        OrderedSet([BitSet(1:2), BitSet(3:4)]),         # Stage 2
    ]
custom_scenariotree = ScenarioTree(stageid_to_scenpart)
## Instantiation of the problem
pb = Problem(
        custom_scenarios,  # scenarios array
        custom_build_subpb,
        custom_probabilities,
        custom_nscenarios,
        custom_nstages,
        custom_stage_to_dim,
        custom_scenariotree
    )
\end{lstlisting}

When the tree to generate is known to be complete, one can fastly generate a scenario tree with the help of the constructor by giving the depth of the tree and the degree of the nodes (assumed to be the same for each node in this case):
\begin{lstlisting}
custom_scenariotree =  ScenarioTree(; depth=custom_nstages, nbranching=2)

pb = Problem(
        custom_scenarios,  # scenarios array
        custom_build_subpb,
        custom_probabilities,
        custom_nscenarios,
        custom_nstages,
        custom_stage_to_dim,
        custom_scenariotree
    )
\end{lstlisting}

The output of the algorithm is the final iterate obtained together with information on the run of the algorithm.  Logs that appear on the console are the input parameters and the functional values obtained along with iterations. If the user wishes to track more information, a callback function can be instantiated and given as an input. This additional information can then either be logged on the console or stored in a dictionary \texttt{hist}.

\begin{lstlisting}
pb = build_simpleexample()
hist=OrderedDict{Symbol, Any}()

y_PH = solve_progressiveHedging(pb, maxiter=150, maxtime=40, hist=hist, callback=callback)
\end{lstlisting}

\end{document}